\documentclass[11pt,letterpaper]{article}
\usepackage{amsthm,color,latexsym,graphicx,url}
\usepackage{overpic}
\usepackage[margin=1in]{geometry}
\usepackage[small,labelfont=bf]{caption}
\usepackage{subcaption}
\usepackage{libertine}\usepackage[libertine]{newtxmath}
\usepackage[scaled=0.96]{zi4}
\usepackage[utf8]{inputenc}
\usepackage{comment}
\usepackage[shortlabels]{enumitem}
\usepackage{xcolor}
\usepackage[colorlinks, linkcolor=purple, citecolor=blue!75!black, urlcolor=purple]{hyperref}

\graphicspath{{./figures/}}

\newcommand*{\NP}[0]{\textrm{NP}}
\newcommand*{\PSPACE}[0]{\textrm{PSPACE}}
\newcommand*{\NPSPACE}[0]{\textrm{NPSPACE}}

\title{Pushing Blocks without Fixed Walls via Checkable Gizmos: \\ Push-1 is \PSPACE-Complete}
\author{%
  MIT Hardness Group%
    \thanks{Artificial first author to highlight that the other authors (in
      alphabetical order) worked as an equal group. Please include all
      authors (including this one) in your bibliography, and refer to the
      authors as “MIT Hardness Group” (without “et al.”).}
\and
  Josh Brunner%
    \thanks{MIT Computer Science and Artificial Intelligence Laboratory,
      32 Vassar St., Cambridge, MA 02139, USA, \protect\url{{brunnerj,lkdc,edemaine,diomidova,della,jaysonl}@mit.edu}}
\and
  Lily Chung\footnotemark[2]
\and
  Erik D. Demaine\footnotemark[2]
\andlinebreak
  Jenny Diomidova\footnotemark[2]
\and
  Della Hendrickson\footnotemark[2]
\and
  Jayson Lynch\footnotemark[2]
}
\date{}

\newif\ifabstract
\abstracttrue
\newif\iffull
\ifabstract \fullfalse \else \fulltrue \fi

\usepackage{hyperref}
\hypersetup{breaklinks,bookmarksnumbered,bookmarksopen,bookmarksopenlevel=2}
{\makeatletter \hypersetup{pdftitle={\@title}}}
{\makeatletter
 \gdef\xxxmark{%
   \expandafter\ifx\csname @mpargs\endcsname\relax %
     \expandafter\ifx\csname @captype\endcsname\relax %
       \marginpar{xxx}%
     \else
       xxx %
     \fi
   \else
     xxx %
   \fi}
 \gdef\xxx{\@ifnextchar[\xxx@lab\xxx@nolab}
 \long\gdef\xxx@lab[#1]#2{\textbf{[\xxxmark #2 ---{\sc #1}]}}
 \long\gdef\xxx@nolab#1{\textbf{[\xxxmark #1]}}
}

{\makeatletter \gdef\fps@figure{!htbp}}

\def\andlinebreak{\end{tabular}\linebreak\begin{tabular}[t]{c}}

\let\realbfseries=\bfseries
\def\bfseries{\realbfseries\boldmath}

\newtheorem{theorem}{Theorem}[section]
\newtheorem{lemma}[theorem]{Lemma}
\newtheorem{corollary}[theorem]{Corollary}
\theoremstyle{definition}
\newtheorem{definition}{Definition}[section]

\let\epsilon=\varepsilon
\def\defn#1{\textbf{\textit{\boldmath #1}}}
\newcommand{\pss}[2]{{#1}^{#2}}
\newcommand{\pssl}[3]{{#1}^{#2}|_{#3}}
\newcommand{\cIn}[0]{c_\text{in}}
\newcommand{\cOut}[0]{c_\text{out}}
\newcommand{\cIni}[1]{c_{\text{in},#1}}
\newcommand{\cOuti}[1]{c_{\text{out},#1}}
\newcommand{\gadget}[1]{\ensuremath{\mathsf{#1}}}
\newcommand{\MSC}[0]{\gadget{MSC}}
\newcommand{\SO}[0]{\gadget{SO}}
\newcommand{\SC}[0]{\gadget{SC}}
\newcommand{\SD}[0]{\gadget{SD}}
\newcommand{\SX}[0]{\gadget{SX}}
\newcommand{\WCX}[0]{\gadget{WCX}}
\newcommand{\SCX}[0]{\gadget{SCX}}
\newcommand{\locs}[1]{L({#1})}

\begin{document}
\maketitle

\begin{abstract}
  We prove \PSPACE-completeness of Push-1: given a rectangular grid of
  $1 \times 1$ cells, each possibly occupied by a movable block,
  can a robot move from one specified location to another,
  given the ability to push up to one block at a time?
  In particular, we remove the need for fixed (immovable) walls
  from a 2022 result.
  This fundamental model of block pushing, introduced in 1999,
  abstracts the mechanics of many video games.
  It was shown \NP-hard in 2000, but its final complexity
  remained open for 25 years.
  Our result uses a new framework for checkable gadgets/gizmos,
  extending a prior framework for checkable gadgets
  to handle reconfiguration problems,
  at the cost of requiring a stronger auxiliary gadget.
  We also introduce a new connection between the motion-planning-through-gadgets framework
  (with an agent) and the Graph Orientation Reconfiguration Problem (with no agent),
  including Nondeterministic Constraint Logic.
\end{abstract}

\section{Introduction}

Countless video games feature \emph{block-pushing puzzles}, where the player
pushes blocks around to achieve some goal.
TV Tropes \cite{tv-tropes} lists over 75 video games and game series
that feature block puzzles,
from famous game series such as
The Legend of Zelda, Pok\'emon, Paper Mario, and Tomb Raider;
to puzzle games such as
Sokoban, Chip's Challenge, Kwirk, Adventures of Lolo, Portal, and Baba Is You;
as well as first-person shooters such as Half-Life,
roguelike games such as NetHack, and
survival horror games such as Resident Evil 2.

\textbf{Push and friends.}
In theoretical computer science, a series of papers over the past 25 years
\cite{ORourke99,Push100,hoffmann-2000-pushstar,PushXCCCG2001,demainepush2F,demaine2003pushing,demaine2004pushpush,PRB16,ani2020pspace,GadgetsChecked_FUN2022}
formalized various pushing-block mechanics
into a family of models collectively called ``Push''.
In all cases, a puzzle consists of an $m \times n$ grid of cells,
where each cell is either empty or contains a $1 \times 1$ movable block,
and a single $1 \times 1$ player/robot/agent moves around the empty cells
via orthogonal (horizontal or vertical) moves.
In \defn{Push-$1$}, the player can walk into a cell containing a block,
provided there is an empty cell on the other side of the block,
in which case both the player and block move one cell in the same direction.
In \defn{Push-$k$}, the player can similarly push up to $k$ blocks at a time,
while in \defn{Push-$*$}, the player can push any number of blocks at a time,
again provided there is an empty cell on the other end of the row of blocks.
In the \defn{PushPush} variation, pushed blocks slide in the direction they are pushed
until they hit another block or wall.
In the \defn{Push-F} variation, some cells are occupied by fixed (immovable) walls;
and in \defn{Push-W}, some edges between cells are also fixed walls;
neither the player nor blocks can move into or through such walls.
The default goal is for the player to get from a specified start location
to a specified goal location, while in \defn{Push-S} variants, the goal
is instead to place the set of $b$ blocks on a specified set of $b$
\emph{storage} locations.
In particular, the famous puzzle video game \defn{Sokoban}
is equivalent to Push-1FS (strength $1$, fixed walls, and storage goal).

Figure~\ref{fig:examples} shows some examples.
In particular, Figure~\ref{fig:successful example}
shows a sequence of moves and three pushes in Push-1
that successfully traverses from the left entrance to the right entrance,
while Figure~\ref{fig:failed example} shows a subsequent sequence of moves
and two pushes (of one block) that fails to traverse from the right entrance
to the left entrance.
In Push-1, the player can never push a block that is contained
in a $2 \times 2$ square of blocks (as observed in \cite{Push100}).
Such blocks are effectively fixed,
making the Push-1 puzzle in Figure~\ref{fig:Push-1 example}
(where in principle all blocks are movable)
equivalent to the Push-1F puzzle in Figure~\ref{fig:Push-1F example}
(where bricked squares are defined to be immovable).

\begin{figure}[t]
  \centering
  \subcaptionbox{\label{fig:Push-1F example} Push-1F puzzle.
  }{\includegraphics[scale=0.6]{figures/svgtiler/no-return}}
  \hfil
  \subcaptionbox{\label{fig:Push-1 example} Equivalent Push-1 puzzle}{\includegraphics[scale=0.6]{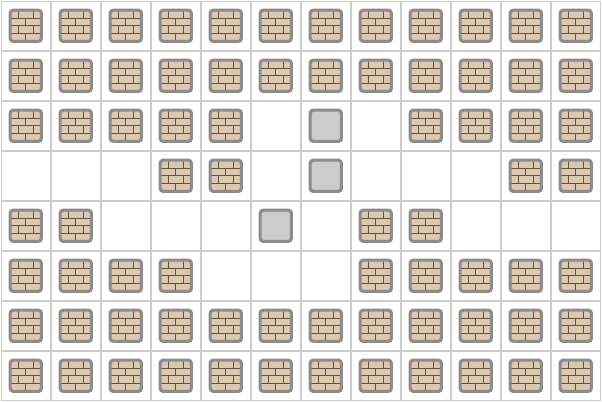}}

  \smallskip
  \subcaptionbox{\label{fig:successful example} Successful forward traversal}{
    \includegraphics[scale=0.3]{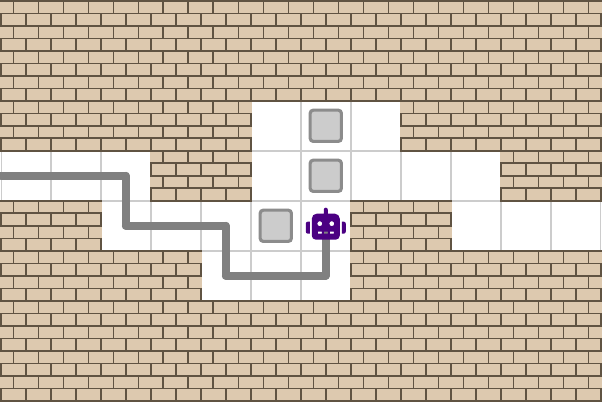}~
    \includegraphics[scale=0.3]{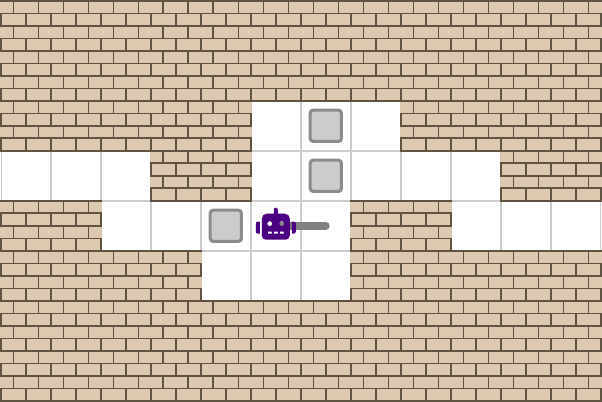}~
    \includegraphics[scale=0.3]{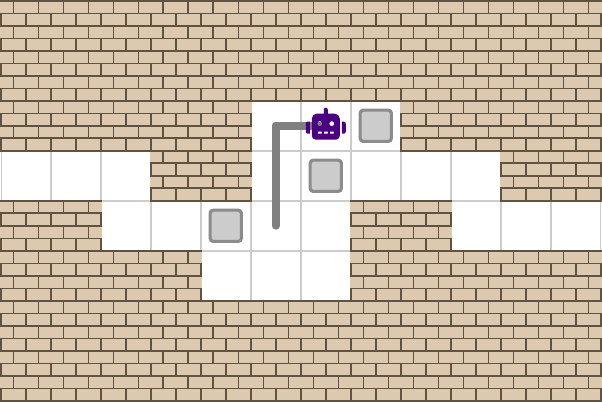}~
    \includegraphics[scale=0.3]{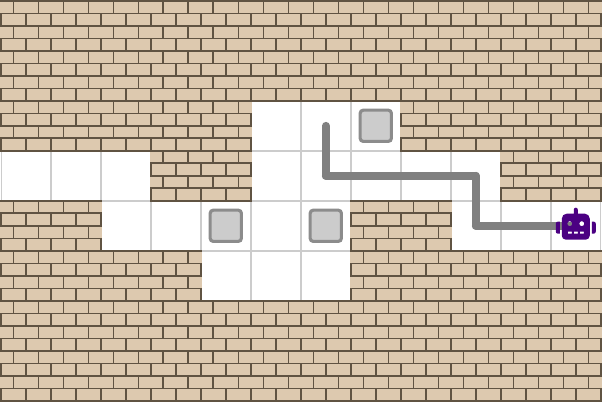}
  }\hfill
  \subcaptionbox{\label{fig:failed example} Attempted backward traversal}
    {\includegraphics[scale=0.3]{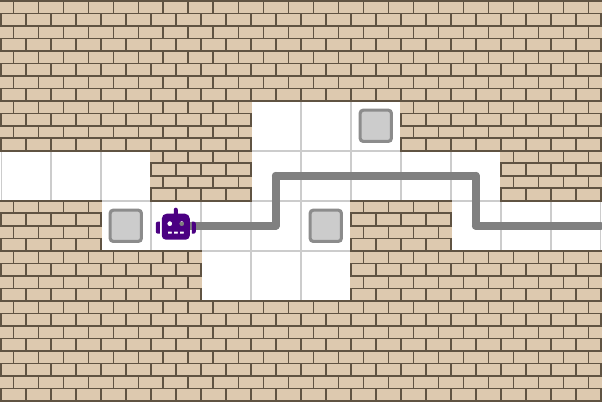}}
  \caption{Example Push-1(F) puzzles: ``no-return'' gadget.
    Bricked squares in (a) are defined to be fixed walls in Push-1F, while
    blocks with brick texture in (b) are effectively fixed in Push-1
    because they are contained in a $2 \times 2$ square of blocks.
    In (c--d) and future figures, to reduce visual clutter,
    we use the brick squares of (a) to draw such effectively fixed blocks
    in Push-1 puzzles, which act like the corresponding Push-1F puzzles.}
  \label{fig:examples}
\end{figure}

We might therefore expect Push-1 and Push-1F to be amenable to similar
complexity analyses.  Indeed, the proofs of \NP-hardness
\cite{Push100,demaine2003pushing} apply equally well to Push-1
and Push-1F, as well as other variants such as PushPush-1 and PushPush-1F.
But \PSPACE-hardness remained open for over two decades.
Along the way, Push-2F was proved \PSPACE-complete \cite{demainepush2F},
as were other variants like PushPush-$k$ \cite{demaine2004pushpush}
and various forms of \emph{pulling} blocks \cite{PRB16,ani2020pspace}.
Finally, a few years ago, Push-1F was proved \PSPACE-complete
\cite{GadgetsChecked_FUN2022}.

Sadly, the Push-1F \PSPACE-hardness proof does not apply to Push-1,
and it seems difficult to adapt the gadgets.
For example, Figure~\ref{fig:Push-1F precursor} shows a central gadget
in the proof.
The functionality of this gadget critically relies on every horizontal push
of a block (such as each push in Figure~\ref{fig:Push-1F precursor steps})
changing which of two incident column paths is blocked.
Given the number of columns that must be tightly packed in a 2D environment
(forcing a rough alternation of up/down directions),
this gadget cannot be directly adapted to Push-1 by thickening
the fixed walls to width~$2$.

\begin{figure}
  \centering
  \subcaptionbox{Gadget}{\includegraphics{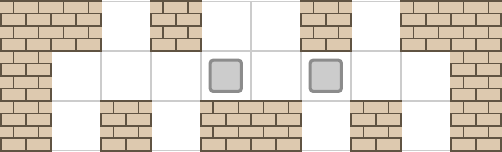}}%
  \hfil
  \subcaptionbox{\label{fig:Push-1F precursor steps}Two intended traversals}{%
    \begin{tabular}{c}
      \includegraphics[scale=0.5]{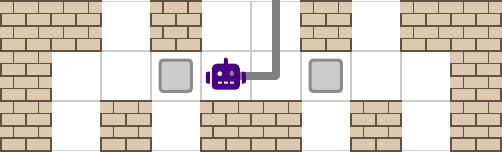}
      \\[\smallskipamount]
      \includegraphics[scale=0.5]{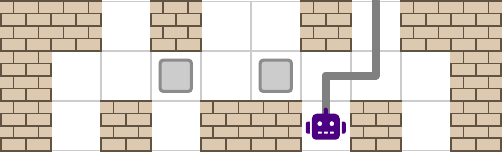}%
      \\[-4pt]
    \end{tabular}%
  }
  \caption{``Checkable proto-precursor'' gadget from the
    Push-1F \PSPACE-hardness proof \cite{GadgetsChecked_FUN2022}.}
  \label{fig:Push-1F precursor}
\end{figure}

In this paper, we prove Push-1 \PSPACE-hard using
a completely different reduction.
The Push-1F \PSPACE-hardness proof \cite{GadgetsChecked_FUN2022}
reduced from motion planning through doors \cite{Doors_FUN2020}
and was inspired by gadgets from a 1998 Sokoban \PSPACE-hardness proof \cite{sokoban}.
Instead, we were inspired by gadgets from a simplified 2005 Sokoban \PSPACE-hardness proof \cite{NCL_TCS,hearn2009games},
which was one of the first applications of Nondeterministic Constraint Logic (NCL) \cite{NCL_TCS,hearn2009games}.
NCL was also previously applied to analyze block-pulling puzzles \cite{ani2020pspace}.
We prove \PSPACE-hardness of Push-1 via a reduction from NCL.

\textbf{Checkable gadgets and gizmos.}
Like the Push-1F \PSPACE-hardness proof \cite{GadgetsChecked_FUN2022},
we use the idea of \defn{checkable gadgets}:
gadgets that can ``break'' from unintended traversals by the player,
and the broken states can later be ruled out by requiring specified ``checking traversals'' to occur at the end of the puzzle.
The checkable gadget framework of \cite{GadgetsChecked_FUN2022} showed how to use a few simple auxiliary gadgets to force these checking traversals to be performed at the end of the puzzle.
This framework effectively allowed the reduction to assume that all gadgets remain unbroken,
reducing each checkable gadget to the subgadget of unbreakable states, in a process called \defn{postselection}.

The fundamental difference between Push-1(F) and Sokoban is that
Push is a \defn{reachability} problem, where the goal is for the player to
reach a specified location, while Sokoban is a \defn{reconfiguration}
problem, where the goal is to reach a particular configuration of the
(unlabeled) blocks.
While the 1998 Sokoban \PSPACE-hardness proof \cite{sokoban} essentially
only used the storage locations to forbid the blocks from reaching certain
broken states (similar to the effect achieved by checkable gadgets \cite{GadgetsChecked_FUN2022}),
the 2005 Sokoban \PSPACE-hardness proof using NCL \cite{NCL_TCS,hearn2009games}
relies more fundamentally on the reconfiguration nature of the storage goal.
In particular,
every location not occupied by a fixed block is easy to reach:
the empty locations remain connected, and
most movable blocks have opposing neighboring empty squares.

Thus we develop a new framework for checkable gadgets, or more precisely,
\defn{checkable gizmos},
which models both reachability and reconfiguration problems.
This framework builds on the \defn{gizmo} framework of
\cite{hendrickson2021gadgets},
originally designed to formalize simulation in the gadget framework,
but which also has the benefit of modeling ``accepting states''
and thus a particular kind of reconfiguration.
By contrast, the checkable gadgets framework of \cite{GadgetsChecked_FUN2022}
required that gadgets stay broken once they break,
preventing a transition from an accepting to nonaccepting state.
Our checkable-gizmos framework allows us to work mainly
with the reconfiguration problem,
and then use a general-purpose reduction from reconfiguration to reachability;
see Section~\ref{sec:checkable-gadgets} for details.
This approach also makes it potentially easier to design checking sequences
for gadgets,
because we can design the checking sequences dependent on the specific
unbroken configuration instead of needing one sequence that works for
all unbroken configurations.

Our checkable gizmos framework obtains stronger results than
the checkable gadgets framework of \cite{GadgetsChecked_FUN2022}.
But it also has a more stringent requirement on which auxiliary gadgets
you need to be able to build for the framework to apply.
Our new framework requires building a gadget called ``single-use closing'' (\SC),
where one path is freely traversable until the player traverses a second disjoint path,
at which point both paths permanently close.
By contrast, the old framework allowed for a weaker form of this gadget called
merged single-use closing (\MSC), where the two paths shared an exit.
Thus both frameworks still have value, depending on which gadgets you can build.
(In fact, when we wrote \cite{GadgetsChecked_FUN2022}, we had not yet built
an \SC{} gadget in Push-1(F), which is why we developed the theory
to allow for an \MSC{} gadget.)
To more clearly distinguish the two checkable frameworks,
we call the old framework \defn{leaky} and the new framework \defn{strict}:
with an \MSC{} gadget, the framework can build a ``weak closing crossover'' that
has limited leakage between two crossing paths, whereas an \SC{} allows it to
build a ``strong closing crossover'' with no such leakage,
enabling stronger guarantees on where the player can go in the checking phase.
Section~\ref{sec:leaky} gives a more detailed comparison.

\textbf{Modeling NCL with gadgets/gizmos.}
To apply this checkable gizmos framework to a reduction from NCL,
we need to be able to represent NCL using gadgets/gizmos.
The challenge is to relate the agentless problem of NCL (where any edge can be flipped ``from outside'')
to the agent-mediated mechanism of gadgets (where the agent can only interact with the gadget at its current location).
We define gizmos that correspond to NCL AND/OR vertices and edges.
More generally, we define a correspondence for any \defn{Graph Orientation
Reconfiguration Problem (GORP)},%
\footnote{Not to be confused with Granola, Oats, Raisins, Peanuts; see~\cite{brown2020gorp}.}
where the goal is to reconfigure a graph from one orientation
into another orientation via a sequence of edge reversals,
subject to certain constraints at vertices.
Each type of vertex specifies which subsets of incident edges, when incoming,
satisfy the vertex; this family must be closed under supersets,
meaning that directing more edges inward can only improve satisfaction,
a property we call \defn{upward-closed}
(from order-theory terminology).
For example, an NCL OR vertex requires that at least one incident edge is
incoming, and an NCL AND vertex requires that either a particular edge is incoming or
two other edges are incoming; both of these properties are upward-closed.
We also define a \defn{GORP crossover} gizmo, which lets the player
reach all vertices in the construction.
See Section~\ref{sec:GORP} for details.

With this technology in place, our Push-1 \PSPACE-hardness proof
``only'' needs to build the following:
\begin{enumerate}
\item GORP gizmos corresponding to NCL AND and OR vertices;
\item a GORP crossover gizmo; and
\item the few simple auxiliary gadgets for the checkable gizmos framework.
  Most of these are the same as in \cite{GadgetsChecked_FUN2022},
  except for replacing \MSC{} with \SC{}.
\end{enumerate}
Section~\ref{sec:Push-1} develops these Push-1 gadgets
and then applies the framework.
A good way to get intuition for the entire proof
is to start by looking at these gadgets,
specifically the GORP gizmos
which intuitively behave like NCL AND and OR vertices,
and then read the other sections to understand exactly
what properties they need to have to make everything work.

Very recently, and independently of our work, DeStefano and Liang
\cite{DeStefano-Liang-2025} also proved Push-1 \PSPACE-complete,
using checkable gadgets and a reduction from motion planning through
(self-closing) doors, closer to the Push-1F proof
\cite{GadgetsChecked_FUN2022}, but introducing the idea of state in the agent.

\section{Gadget Framework}

In this section, we review the relevant parts of the motion-planning-through-gadgets framework \cite{Toggles_FUN2018, demaine2020toward}.
We will write formal definitions in the language of \emph{gizmos},
a useful abstraction of gadgets introduced in \cite{hendrickson2021gadgets}.
This will be helpful later when discussing the particular details of checkable gadgets.
For now, we informally review the broad ideas of the gadget framework.
Unless stated otherwise, all reductions in this paper are polynomial-time many-one reductions.

The idea of the gadget framework is to represent a motion-planning problem as a network of ``gadgets''
through which an agent can move.
Each gadget has several external \emph{locations},
and edges of the network connect gadgets to each other by linking their locations.
The agent can freely travel along these network edges,
but the agent's ability to traverse the gadgets themselves is restricted to traversals allowed by the gadget.
In addition, making such gadget traversals may change the state of the gadget,
so that certain traversals within a gadget may be available or unavailable
depending on which traversals have already been made on that gadget.
A gadget can be specified by a \defn{state diagram}, which details the traversals allowed in each of the possible states of the gadget, annotated with the state that results from each traversal.
An example of a gadget is the \defn{locking 2-toggle},
whose state diagram is shown in Figure~\ref{fig:state-diagram}.

\begin{figure}[htp]
	\centering
	\includegraphics[scale=0.75]{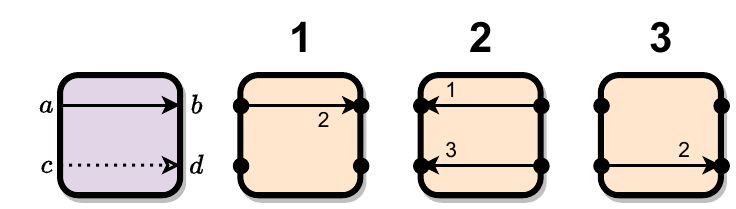}
	\caption{Icon (left) and state diagram (right three frames) of a locking 2-toggle. This gadget has four locations $a, b, c, d$ and three states (numbered at top).  Arrows in the state diagram indicate allowed traversals, labeled with the resulting state: for instance, the arrow in frame 1 indicates that in state 1, the traversal $a \to b$ is allowed and transitions the gadget to state 2.}
	\label{fig:state-diagram}
\end{figure}

We are mainly interested in two types of problems:
\emph{reachability}, which asks just whether the agent can travel through the network from a start location \(s\) to a target location \(t\);
and \emph{targeted reconfiguration}, which asks whether the agent can do the same while also leaving the gadgets in certain specified states.  A goal of the gadget framework is to characterize the computational complexity of these problems on various types of gadgets.  For instance, recalling the previous example, it has been shown that both the reachability and targeted reconfiguration problems on locking 2-toggles are \PSPACE-complete \cite{demaine2020toward,GadgetsVictory2023}.

Such results are useful starting points for proving computational hardness of various types of games and puzzles.
The idea is to construct various gadgets in the game or puzzle under consideration, and then to combine these gadgets into a network that ``simulates'' a known computationally hard gadget, thus proving the original game or puzzle hard as well.
Often, constructing these networks in the original game or puzzle imposes additional constraints, the most common of which is planarity of the network.

\subsection{Gizmos}

When building reductions out of gadgets, we often do not care about the precise states of the gadgets, but only about which sequences of traversals are possible. Hendrickson \cite{hendrickson2021gadgets} introduced the abstraction of \defn{gizmos} to capture this idea, and to formalize the notion of simulation in the gadget framework.
The relationship between gadgets and gizmos is closely analogous to the relationship between state machines and the formal languages they accept.  For our purposes, the gizmo view is useful because it allows us to cleanly define a gizmo that captures the behavior of an NCL vertex, while allowing us to ignore the precise states of the gadget, which in our Push-1 implementation are more complicated to define and work with.

Given a set \(L\) of \defn{locations}, a \defn{traversal} on \(L\) is a pair \(a \to b\) where \(a, b \in L\); and a \defn{traversal sequence} on \(L\) is a sequence of such traversals \([a_1 \to b_1, \dots, a_k \to b_k]\).
A gizmo is a language of ``allowed'' traversal sequences, which must obey two rules:
(1)~it is always possible to perform a trivial traversal which enters and immediately exits at the same location;
and (2)~a pair of consecutive traversals $a \to b$ followed by $b \to c$ can be viewed as a single longer traversal $a \to c$.

\begin{definition}[\cite{hendrickson2021gadgets}]
  \label{defn:gizmo}
  A \defn{gizmo} \(G\) on a location set \(\locs{G}\) is a set of traversal sequences on \(\locs{G}\) satisfying the following properties,
  where \(X\) and \(Y\) denote arbitrary traversal sequences on \(\locs{G}\)
  and $XY$ denotes the operation of concatenating two traversal sequences to form a new traversal sequence.
  \begin{enumerate}
  \item If \(XY \in G\), then \(X[a \to a]Y \in G\) for any location \(a \in \locs{G}\).
  \item If \(X[a \to b, b \to c]Y \in G\), then \(X[a \to c]Y \in G\).
  \end{enumerate}
\end{definition}

Gizmos can be connected together in a network to form a single combined gizmo called a \emph{simulation}.
The network connects together locations of the individual gizmos
via an equivalence relation $\sim$,
intuitively allowing the agent to traverse between equivalent locations,
and specifies a subset $L$ of equivalence classes as ``exposed locations'' in the combined gizmo.
The intended behavior is that a traversal between two exposed locations of the simulation consists of a walk through the network,
possibly passing through multiple individual gizmos.
A traversal sequence is allowed by the simulation exactly when each individual traversal can be realized by such a walk,
and the subsequence of traversals through each individual gizmo is allowed by that gizmo.
The following definition formalizes this idea:

\begin{definition}
  \label{defn:network}
  \label{defn:simulation}
  Let $\{G_i\}_{i \in I}$ be a collection of gizmos,
  $\sim$ be an equivalence relation on the disjoint union
  $\bigsqcup_{i \in I} \locs{G_i}$ of their locations,
  and $L$ be a subset of the equivalence classes.
  Then $(\{G_i\}_{i \in I}, \sim, L)$ forms a \defn{simulation}, which describes a gizmo on $L$ defined as follows.

  Let $X = [a_1 \to b_1, \dots, a_n \to b_n]$ be a traversal sequence with $a_j, b_j \in L$.
  The simulated gizmo contains $X$ if and only if there exist corresponding traversal sequences $Y_1,\dots,Y_n$
  satisfying the following:
  \begin{enumerate}
  \item Each $Y_j$ is a sequence of traversals
    $[c_1 \to d_1, \dots, c_m \to d_m]$
    such that
    \begin{enumerate}
    \item every traversal $c_k \to d_k$ is a traversal on $\locs{G_i}$ for some $i \in I$;
    \item $c_1 \in a_j$ and $d_m \in b_j$; and
    \item $d_k \sim c_{k+1}$ for each $1 \le k < m$.
    \end{enumerate}
  \item For each $i \in I$, the subsequence $Z_i$ of $Y_1 Y_2 \cdots Y_n$ consisting of traversals on $\locs{G_i}$ satisfies $Z_i \in G_i$.
  \end{enumerate}
\end{definition}

A simple example of gizmo simulation is using two diodes (gizmos that allow any number of traversals in only one direction $a \to b$) to simulate a wire (a gizmo that allows an arbitrary sequence of traversals in both directions $a \leftrightarrow b$). The simulation consists of two diodes connected in parallel but facing opposite ways, with the exposed locations being those two connections. The resulting simulated gizmo allows traversal in both directions, as desired.

In this paper, we consider only simulations where the location sets and the gizmo collection $\{G_i\}_{i \in I}$ are finite.

A set $S$ of gizmos \defn{simulates} a gizmo $H$ if $H$ is equivalent (up to relabeling locations) to the gizmo described by some simulation $(\{G_i\}_{i \in I}, \sim, L)$ whose $G_i$ are all gizmos from $S$.\footnote{%
This definition of simulation is equivalent to the definition in \cite{hendrickson2021gadgets}, which defines a simulation as a combination of the tensor product, quotient (Definition~\ref{def:quotient}), and subgizmo (Definition~\ref{def:subgizmo}) operations.  A tensor product is a simulation of the form $(G_i, =, \bigsqcup_i L(G_i))$, where $\bigsqcup$ denotes disjoint union.%
}
Simulations can be composed: if $S, S'$ are sets of gizmos, $S$ simulates every gizmo in $S'$, and $S'$ simulates $H$, then $S$ also simulates $H$.

Two important special cases of simulations are \emph{quotients} and \emph{subgizmos}.

\begin{definition}
  \label{def:quotient}
  Let $G$ be a gizmo and $\sim$ be an equivalence relation on $\locs{G}$.
  The \defn{quotient gizmo} $G/\sim$ is defined by the one-gizmo simulation $(G, \sim, \locs{G}/\sim)$.
\end{definition}
\begin{definition}
  \label{def:subgizmo}
  Let $G$ be a gizmo and $L \subset \locs{G}$.
  The \defn{subgizmo} $G|_L$ is defined by the one-gizmo simulation $(G, =, L)$, where $=$ is the minimal equivalence relation.
\end{definition}

\subsection{Planar Gizmos}

Push-1 is a two-dimensional problem, so we need to make sure our simulations work in a planar environment.
We introduce new definitions for planar gizmos and simulations,
following past work on planar gadgets \cite{Toggles_FUN2018,demaine2020toward,Doors_FUN2020,GadgetsChecked_FUN2022}.

A \defn{planar gizmo} is a gizmo together with a cyclic order on its location set.
We define planar simulation by requiring the graph of connections in the network to be planar.
To encode the exposed locations in the same drawing, we add one extra vertex $\infty$ representing the exterior of the simulation, with one edge to $\infty$ for each exposed location.

\begin{definition} \label{defn:planar-simulation}
  Given a simulation $(\{G_i\}_{i \in I}, \sim, L)$
  from Definition~\ref{defn:network},
  consider the multigraph that has a vertex for each gizmo $G_i$, a vertex for each equivalence class of $\sim$, and one additional vertex~$\infty$.
  For each location $a \in \locs{G_i}$, there is an edge from $G_i$ to the equivalence class of $a$.
  There is also an edge $(\ell, \infty)$ for each $\ell \in L$.
  A \defn{planar simulation} consists of a simulation
  $(\{G_i\}_{i \in I}, \sim, L)$ where each $G_i$ is a planar gizmo,
  together with a planar embedding of this multigraph
  such that the ordering of the edges incident to each $G_i$ matches the cyclic order of the locations of $G_i$.
  The cyclic order of the simulation's locations $L$ is defined to be the reverse\footnote{The ordering is reversed because $\infty$ represents the ``exterior'' of the simulation.  This distinction is important when planar gizmos may be rotated but not reflected.  Because Push-1 is symmetric under reflections, we assume in this paper that the reflection of a planar gizmo can be obtained whenever needed, so this will not be critical.} of the ordering of the edges incident to~$\infty$.
\end{definition}

A set $S$ of gizmos \defn{planarly simulates} a gizmo $H$ if $H$ is equivalent (up to relabeling locations, respecting cyclic order) to the gizmo described by some planar simulation that uses only gizmos from $S$.

\subsection{States, Reachability, and Reconfiguration}

A \defn{regular} gizmo is a set of traversal sequences that is a regular language; that is, it is recognized by a finite automaton.
Regular gizmos are essentially equivalent\footnote{More precisely, a regular gizmo is a gadget together with a specified starting state and a specified set of ``accepting'' states.} to \emph{gadgets} as defined in earlier work, and all the gizmos we explicitly construct in this paper will be regular.
They can be specified by means of a \emph{state diagram},
such as the earlier example in Figure~\ref{fig:state-diagram}.

Next we define a reconfiguration decision problem for gizmos.
In the case of regular gizmos, this problem asks: given a network of regular gizmos with specified locations $s$ and $t$,
can an agent travel from $s$ to $t$ through the network, while leaving all of the finite automata in accepting states?
Formally:
\begin{definition}
  Let $S$ be a finite set of gizmos.
  The \defn{[planar] targeted set reconfiguration problem} on $S$ is the following decision problem:
  given a [planar] simulation of a gizmo $G$ on the location set $\{s, t\}$ using gizmos from $S$, is $[s \to t] \in G$?
\end{definition}
Targeted set reconfiguration is the only kind of gadget reconfiguration we consider in this paper,
so we will sometimes call it simply ``reconfiguration''.

If all gizmos $H \in S$ are \defn{prefix-closed}, meaning that $XY \in H$ implies $X \in H$, then we call this a \defn{reachability} problem.
A prefix-closed regular gizmo is described by a finite automaton whose states are all accepting, so the goal here is just to make it from $s$ to $t$ without any further constraints.
The main result of Section~\ref{sec:checkable-gadgets} will be a general technique for reducing from targeted set reconfiguration problems to reachability problems.

\section{Checkable Gizmos/Gadgets Frameworks}
\label{sec:checkable-gadgets}

In this section, we adapt the checkable gadgets framework from \cite{GadgetsChecked_FUN2022} to handle reconfiguration problems.
The main idea of our framework is to use a sequence of forced traversals at the very end of a solution to ensure that every gizmo was left in a desired final state.
This is what will allow us to reduce reconfiguration problems to reachability problems,
because we can ensure that the goal location is reachable only by performing these forced traversals,
which check that the gadgets were left in the required final configuration.

The checkable gadgets framework from \cite{GadgetsChecked_FUN2022} partially accomplished this goal.
In that framework, one can designate a certain set of gadget states as \emph{broken},
and use machinery built from certain base gadgets
to ensure that all gadgets were left in unbroken states.
For this to work, it was required that it be impossible to transition from a broken state to an unbroken state.\footnote{This notion of checkable gadgets is also closely related to the ``verified gadgets'' in \cite{jaysonthesis}.  The main difference is that checking traversals are allowed to use locations of the gadget being checked, rather than only using separate ``verifying locations''.}
Unfortunately, this requirement makes the framework unsuitable for reconfiguration problems,
where it is possible for a gadget to transition many times between the desired final state
and other states that are not final but are still part of the intended operation of the gadget.

We will show how to remove this limitation and give a version of the framework that works for reconfiguration problems.
This change will have a price: one of the necessary base gadgets becomes harder to build.
This means there are really \emph{two} checkable gizmo frameworks.
The \emph{strict} checkable gizmo framework (our new framework) gives stronger guarantees that are useful for reconfiguration problems, but requires a more robust base gadget.
The \emph{leaky} checkable gizmo framework (the previous framework) has weaker guarantees but lifts this requirement on the base gadget.
In this section, we will focus on the strict checkable gizmo framework, but will briefly outline the leaky version in Section~\ref{sec:leaky}.

\subsection{Checking and Postselection}

\emph{Postselection} captures the effect of a final ``checking'' traversal sequence~$C$:
we keep exactly the traversal sequences after which $C$ can still be completed.
\begin{definition}
  Let \(G\) be a gizmo and \(C\) be a traversal sequence on \(\locs{G}\).
  The \defn{postselection} \(\pss{G}{C}\) is the set of traversal sequences
  $X$ on $\locs{G}$ such that $XC\in G$.
  That is, a traversal sequence \(X\) is permitted by \(\pss{G}{C}\) if and only if
  the ``checking'' sequence \(C\) would be possible after performing the same traversal sequence in \(G\).
  It is easy to check that $\pss{G}{C}$ is a gizmo.
\end{definition}

Suppose we have a gizmo \(H\) that can be obtained from \(G\) by postselecting on some checking sequence~\(C\),
and then possibly [planarly] identifying and/or closing some locations (via quotient and/or subgizmo operations).
In this case, we say that \(G\) is a \defn{[planarly] checkable} \(H\).%
\footnote{This definition differs slightly from the definition of ``checkable'' in \cite{GadgetsChecked_FUN2022}.  That paper permits closing locations (subgizmo operations) in the definition, but not identifications (quotient operations).}
Formally this means that there exists a checking traversal sequence \(C\) and a [planar] simulation of $H$ using only the single gizmo~$\pss{G}{C}$.
Practically it means that, if we can force the checking sequence $C$ at the end, then $G$ acts like~$H$.

\subsection{Nonlocal Simulation}

Forcing checking traversals to occur at the end
requires a global modification to the gizmo network,
which cannot be captured by the local notions of simulation from
Definitions~\ref{defn:simulation} and~\ref{defn:planar-simulation}.
Thus we introduce a more general notion of simulation among gizmos:

\begin{definition}[\cite{GadgetsChecked_FUN2022}]
  A set $S$ of gizmos \defn{[planarly] nonlocally simulates} a gizmo $H$ if,
  for every set of gizmos $T$,
  there is a polynomial-time reduction
  from [planar] targeted set reconfiguration on $\{H\} \cup T$
  to [planar] targeted set reconfiguration on $S \cup T$.
\end{definition}

(Local) simulations are a special case of nonlocal simulations:
a reduction simply replaces each instance of the simulated gadget
with the corresponding simulation network.

We can combine multiple nonlocal simulations together as follows:

\begin{lemma}
  \label{lem:nonlocal set}
  Let $S_1, S_2, T$ be sets of gizmos such that $S_2$ is finite.
  If $S_1$ [planarly] nonlocally simulates every gizmo in $S_2$,
  then there is a polynomial-time reduction from [planar] targeted set reconfiguration on $S_2 \cup T$ to [planar] targeted set reconfiguration on $S_1 \cup T$.
\end{lemma}

\begin{proof}
  Let $S_2 = \{G_1, \dots, G_n\}$ and define $G_{:i}$ to be the prefix $\{G_1, \dots, G_i\}$.
  Then there is a chain of $n$ reductions between targeted set reconfiguration problems
  \[S_2 \cup T = G_{:n} \cup T \underset{1}{\leadsto} S_1 \cup G_{:n-1} \cup T \underset{2}{\leadsto} S_1 \cup G_{:n-2} \cup T \underset{3}{\leadsto} \cdots \underset{n-1}{\leadsto} S_1 \cup G_{:1} \cup T \underset{n}{\leadsto} S_1 \cup T,\]
  where reduction $\underset{i}{\leadsto}$ is an application of the given nonlocal simulation of~$G_i$.
\end{proof}

\begin{corollary}
  \label{lem:nonlocal composition}
  Let $S_1, S_2$ be sets of gizmos such that $S_2$ is finite, and let $H$ be a gizmo.
  If $S_1$ [planarly] nonlocally simulates every gizmo in $S_2$, and $S_2$ [planarly] nonlocally simulates $H$,
  then $S_1$ [planarly] nonlocally simulates~$H$.
\end{corollary}
\begin{proof}
  By definition of nonlocal simulation and Lemma~\ref{lem:nonlocal set}, there are reductions between targeted set reconfiguration problems $\{H\} \cup T \leadsto S_2 \cup T \leadsto S_1 \cup T$.
\end{proof}

\subsection{Main Theorem}

The point of the checkable gizmo framework is that,
if we can build two simple base gadgets $\SO$ (single-use opening) and $\SC$ (single-use closing) in a game or puzzle,
then we can force any desired checking sequence $C$ to occur at the end of the puzzle.
This guarantee ensures that, in any overall solution, the gadgets are left in configurations from which $C$ can be performed,
which is exactly the behavior captured by postselection $G^C$.
Formally:

\begin{theorem}
  \label{thm:check}
  Let $G$ be a gizmo and $C$ be a traversal sequence on $\locs{G}$.
  Then $\{G, \SO, \SC\}$ planarly nonlocally simulates $\pss{G}{C}$.
\end{theorem}

Theorem~\ref{thm:check} is similar to \cite[Theorem~1]{GadgetsChecked_FUN2022}.
The key difference is that $\pss{G}{C}$ may not be prefix-closed even if $G$ is
prefix-closed (as illustrated in Figure~\ref{fig:sscd}),
whereas the unbroken states in \cite{GadgetsChecked_FUN2022} require prefix-closure.
This generalization is what allows Theorem~\ref{thm:check} to reduce the \emph{targeted set reconfiguration} problem on $\pss{G}{C}$ to a \emph{reachability} problem on $\{G, \SO, \SC\}$.

\begin{figure}[htp]
  \centering
  \includegraphics[scale=0.75]{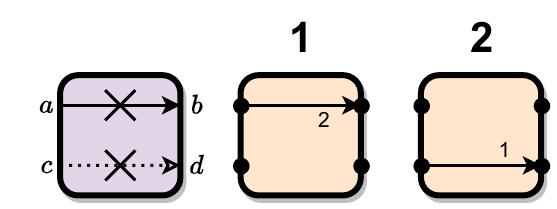}
  \caption{Icon and state diagram for a symmetric self-closing door from \cite{Doors_FUN2020}.
    Crosses indicate traversals that sometimes close,
    depending on the state of the door, while the dotted line shows the
    currently closed traversal.
    The gizmo $G$ corresponding to state $1$ (with both states accepting) is prefix-closed.
    However, $\pss{G}{[c\to d]}$ is not prefix-closed,
    because $G$ contains $[a\to b,c\to d]$ but not $[c\to d]$.
    In fact, $\pss{G}{[c\to d]}$ is the gizmo corresponding to state $1$
    of the symmetric self-closing door where only state $2$ is accepting.
  }
  \label{fig:sscd}
\end{figure}

\subsection{Base Gadgets}

\begin{figure}
  \centering
  \def\scale{0.8}
  \def\width{2.9cm}
  \subcaptionbox{\label{fig:SO_state}\centering Single-use opening (\SO)}[\width]{\includegraphics[scale=\scale]{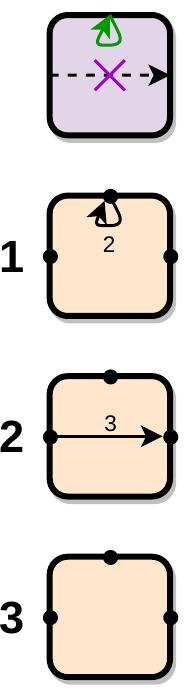}}
  \hspace{2cm}
  \subcaptionbox{\label{fig:SC_state}\centering Single-use closing (\SC)}[\width]{\includegraphics[scale=\scale]{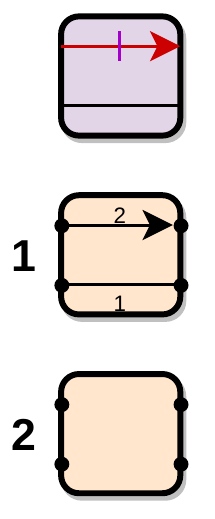}}
  \hspace{2cm}
  \subcaptionbox{\label{fig:MSC_state}\centering Merged single-use closing (\MSC)}[\width]{\includegraphics[scale=\scale]{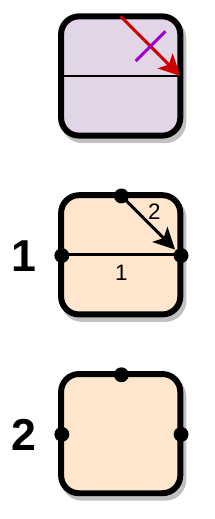}}
  \caption{Icons (top) and state diagrams (bottom) for three basic gadgets.
    Green arrows show opening traversals (traversals that make more traversals available),
    red arrows show closing traversals (traversals that destroy a previously available traversal),
    and purple crosses indicate traversals that permanently close themselves.}
  \label{fig:base gadgets}
\end{figure}

\begin{figure}
  \centering
  \def\scale{0.8}
  \def\width{2.9cm}
  \subcaptionbox{\label{fig:SD_state}\centering Dicrumbler (\SD)}[\width]{\includegraphics[scale=\scale]{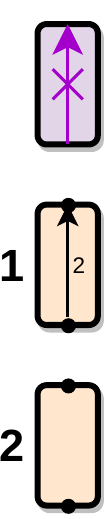}}
  \hfill
  \subcaptionbox{\label{fig:SX_state}\centering Single-use crossover (\SX)}[\width]{\includegraphics[scale=\scale]{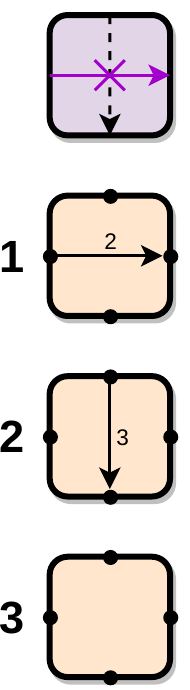}}
  \hfill
  \subcaptionbox{\label{fig:WCX_state}\centering Weak closing crossover (\WCX)}[\width]{\includegraphics[scale=\scale]{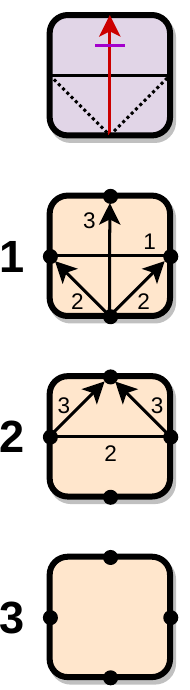}}
  \hfill
  \subcaptionbox{\label{fig:SCX_state}\centering Strong closing crossover (\SCX)}[\width]{\includegraphics[scale=\scale]{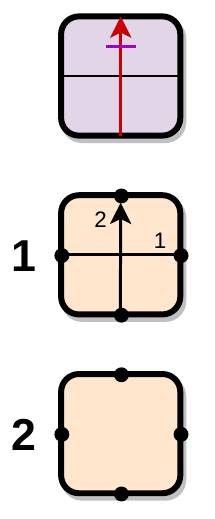}}
  \caption{Icons (top) and state diagrams (bottom) for four derived gadgets.
    Red arrows show closing traversals
    and purple crosses indicate traversals that permanently close themselves.}
  \label{fig:derived gadgets}
\end{figure}

Figures~\ref{fig:SO_state} and~\ref{fig:SC_state} show
the two base gadgets we use to implement checking.
The \defn{single-use opening} (\defn{\SO}) gadget is the same as in \cite{GadgetsChecked_FUN2022},
while the \defn{single-use closing} (\defn{\SC{}}) gadget is a new two-state four-location gadget.
It initially allows horizontal traversals along the bottom tunnel.
The top tunnel can be traversed from left to right exactly once,
after which no traversals are possible.

If we merge the two rightmost locations of \SC{}, we obtain the \defn{merged single-use closing} (\defn{\MSC{}}) gadget, shown in Figure~\ref{fig:MSC_state}.
The previous framework \cite{GadgetsChecked_FUN2022} used \MSC{} in place of \SC,
and showed that $\{\SO{}, \MSC{}\}$ planarly simulates the \defn{dicrumbler} (\defn{\SD}), \defn{single-use crossover} (\defn{\SX}), and \defn{weak closing crossover} (\defn{\WCX}) gadgets shown in Figure~\ref{fig:derived gadgets}.
Because \SC{} simulates \MSC{}, $\{\SO{}, \SC{}\}$ also simulates these gadgets,
so we can use \SD{}, \SX{}, and \WCX{} as building blocks
in our new framework as well.

By combining \WCX{} and \SC{} as shown in Figure~\ref{fig:strong-closing-crossover}, we can obtain the \defn{strong closing crossover} (\defn{\SCX}) shown in Figure~\ref{fig:SCX_state}.
This gadget is exactly the same as \SC{} except for the cyclic ordering of its locations, which has the two tunnels cross each other.
Compared to \WCX{}, \SCX{} prevents the agent from ``leaking'' from the vertical tunnel into the horizontal tunnel.
This property will be critical for our stronger checking framework.

\begin{figure}[htp]
	\centering
	\includegraphics[scale=0.75]{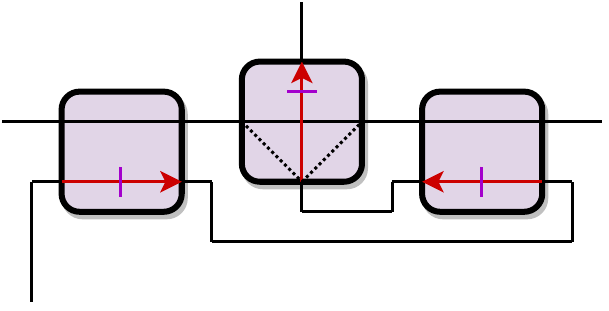}
	\caption{Construction of a strong closing crossover (\SCX) using \SC{} and \WCX{}.}
	\label{fig:strong-closing-crossover}
\end{figure}

\subsection{Proof of Main Theorem}

Now we sketch the proof of Theorem~\ref{thm:check} along the same lines as \cite{GadgetsChecked_FUN2022}.
The differences come from replacing uses of \WCX{} with \SCX, which prevents the agent from making certain unintended traversals (``leaks'').

The easiest case is when the checking phase is just one forced traversal from a special location $\cIn$ to a special location $\cOut$,
where these two locations are used only for checking and are not part of the gadget $H$ being simulated.
We call this weaker notion ``simple checkability''.
We will prove Theorem~\ref{thm:check} first for simply checkable gizmos,
and then show how to convert any checkable gizmo into a simply checkable one.

\begin{definition}
  Let $G, H$ be gizmos with $\locs{G} = \locs{H} \sqcup \{\cIn, \cOut\}$.
  For planar gizmos, this decomposition must respect the cyclic order.
  Then $G$ is a \defn{simply checkable} \(H\)
  if it satisfies the following conditions.%
  \begin{enumerate}
  \item If \(X \in H\), then \(X[\cIn \to \cOut] \in G\).
  \item For any traversal sequence \(Y \in G\),
    let \(\widetilde Y\) be the traversal sequence obtained from $Y$
    by removing any instances of the trivial traversals
    \(\cIn \to \cIn\) or \(\cOut \to \cOut\).
    Then either $\widetilde Y$ contains only traversals on $\locs{H}$ or $\widetilde Y = X[\cIn \to \cOut]$ for some $X \in H$.
  \end{enumerate}
\end{definition}

In particular, it follows from this definition that $H = \pssl{G}{[\cIn \to \cOut]}{\locs{H}}$
(where $G|_L$ denotes a subgizmo from Definition~\ref{def:subgizmo}),
so a simply checkable \(H\) is also a planarly checkable \(H\).

\begin{lemma}
  \label{lem:simply checkable}
  Let $G$ be a simply checkable $H$.
  Then $\{G\}$ nonlocally simulates $H$,
  and $\{G, \SCX\}$ planarly nonlocally simulates $H$.
\end{lemma}
\begin{proof}
  Refer to Figure~\ref{fig:simply checkable}.
  Suppose we are given
  an instance of targeted set reconfiguration on $\{H\} \cup T$, i.e.,
  a simulation of a gizmo on $\{s, t\}$ using gizmos from $\{H\} \cup T$.
  We construct an instance of targeted set reconfiguration on $\{G\} \cup T$,
  i.e., a simulation of a gizmo on $\{s, t'\}$ using gizmos from $\{G\} \cup T$,
  as follows.
  Replace every instance of $H$ with an instance of $G$.
  Label these instances $g_1, \dots, g_n$,
  and label the $\cIn$ and $\cOut$ locations of each $g_i$ as
  $\cIni{i}$ and $\cOuti{i}$ respectively.
  Identify $\cOuti{i}$ with $\cIni{i+1}$ for $1 \leq i < n$.
  Additionally identify the old target location $t$ with $\cIni{1}$,
  and identify the new target location $t'$ with $\cOuti{n}$.

\begin{figure}
  \centering
  \begin{subfigure}{0.47\linewidth}
    \centering
    \begin{overpic}[scale=0.7]{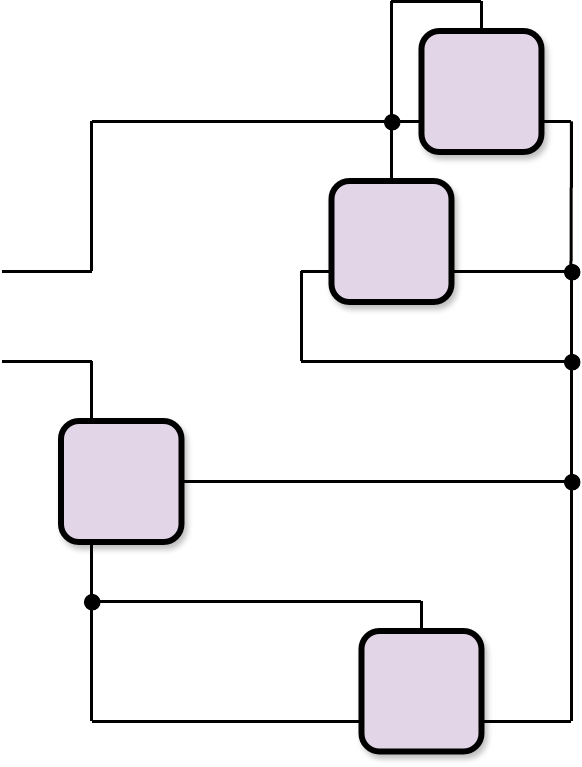}
      \put(2,54){\makebox(0,0)[c]{$s$}}
      \put(2,66){\makebox(0,0)[c]{$t$}}
      \put(63,88){\makebox(0,0)[c]{\Large $g_4$}}
      \put(51,68){\makebox(0,0)[c]{\Large $g_1$}}
      \put(16,37){\makebox(0,0)[c]{\Large $g_3$}}
      \put(55,9){\makebox(0,0)[c]{\Large $g_2$}}
    \end{overpic}
    \caption{Given (old) simulation: can you get from $s$ to $t$?}
  \end{subfigure}%
  \hfill
  \begin{subfigure}{0.47\linewidth}
    \centering
    \begin{overpic}[scale=0.7]{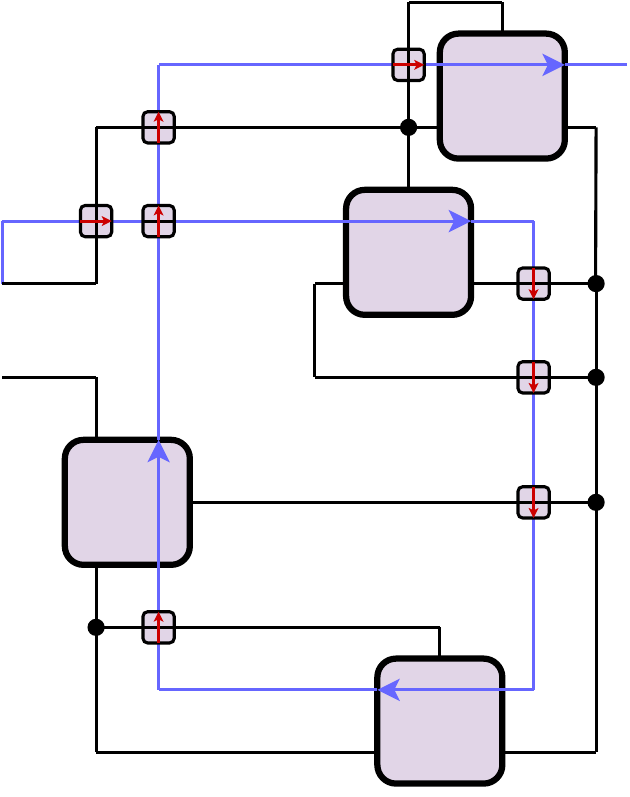}
      \put(2,54){\makebox(0,0)[c]{$s$}}
      \put(2,66){\makebox(0,0)[c]{$t$}}
      \put(77,94){\makebox(0,0)[c]{$t'$}}
      \put(63,88){\makebox(0,0)[c]{\Large $g_4$}}
      \put(51,68){\makebox(0,0)[c]{\Large $g_1$}}
      \put(16,37){\makebox(0,0)[c]{\Large $g_3$}}
      \put(55,9){\makebox(0,0)[c]{\Large $g_2$}}
    \end{overpic}
    \caption{Augmented (new) simulation: can you get from $s$ to $t'$?}
  \end{subfigure}
  \caption{Nonlocal simulation for the proof of Lemma~\ref{lem:simply checkable}, adapted from \cite[Figure~21]{GadgetsChecked_FUN2022} using \SCX{} gadgets instead of \WCX{} gadgets.}
  \label{fig:simply checkable}
\end{figure}

  It remains to show that the new simulation has a solution if and only if the old simulation does.

  Given a solution to the old simulation, we can follow the same solution to get from $s$ to $t$,
  then follow the chain of traversals
  $$t \sim \cIni{1} \to \cOuti{1} \sim \cIni{2} \to \cOuti{2} \sim \cIni{3} \to \cdots \to \cOuti{n} \sim t'$$
  to reach~$t'$.
  The resulting path satisfies each gizmo $g_i$ by definition of simply checkable.

  Conversely, suppose there is a solution to the new simulation.
  The only way to reach $t'$ is by traversing $\cIni{n} \to \cOuti{n}$, so the solution must end with this traversal (ignoring trivial $\cIn \to \cIn$ and $\cOut \to \cOut$ traversals).
  Similarly, the only way to reach each $\cIni{i+1}$ is by traversing $\cIni{i} \to \cOuti{i}$.
  By induction backwards on $i$, the solution must end with the traversals
  $[\cIni{1} \to \cOuti{1}, \dots, \cIni{n} \to \cOuti{n}]$
  (again ignoring trivial traversals).
  By the definition of simply checkable,
  the subsequence of nontrivial traversals performed on each $g_i$
  is of the form $X[\cIni{i} \to \cOuti{i}]$
  for some $X \in H$.
  The only way to reach $\cIni{1}$ is by reaching~$t$
  (or another location in the equivalence class).
  Because all of the $\cIn \to \cOut$ traversals occur at the end,
  the remaining prefix of traversals must form a valid solution to the old simulation.

  In the planar case, the identifications cannot necessarily be made in a planar way, so we must replace each edge crossing with an \SCX{} gadget, as in Figure~\ref{fig:simply checkable}.  The rest of the argument is identical.
\end{proof}

Finally we can prove the main theorem of our framework, which just requires a single additional construction.

\begin{proof}[Proof of Theorem~\ref{thm:check}]
  By Lemma~\ref{lem:simply checkable}, it suffices to simulate a simply checkable $\pss{G}{C}$ using the gizmos $\{G, \SO, \allowbreak \SC\}$.
  Write $C = [a_1 \to b_1, \dots, a_k \to b_k]$.
  Figure~\ref{fig:simply checkable postselect} shows the simulation. It starts by adding a path from $\cIn$ which crosses all of the paths leading to locations of $G$ using \SCX{} gadgets, which will prevent the agent from leaving and making any traversals other than the intended checking traversals $a_i \to b_i$ in the checking phase.
  Then, for each $i$, it adds an \SD{} $D_i$ gadget that leads to
  location $a_i$, and an \SO{} $O_i$ gadget coming from location~$b_i$.
  Before we enter the checking phase by entering $\cIn$,
  neither of these gadgets can be used.
  In the checking phase, they allow the agent to make the traversal $a_i \to b_i$,
  and the \SCX{} gadgets above will prevent any other traversals. These $O_i$ and $D_i$ gadgets are connected together in a long chain to ensure that the checking traversals $a_i \to b_i$ are performed in the correct order. These connections may cross each other as well as the exit lines from the gadget to the outside world, but we can use \SX{} and \SCX{} gadgets to handle these crossings without allowing any unintended traversals.

\begin{figure}[t]
  \centering
  \begin{overpic}[width=0.7\linewidth]{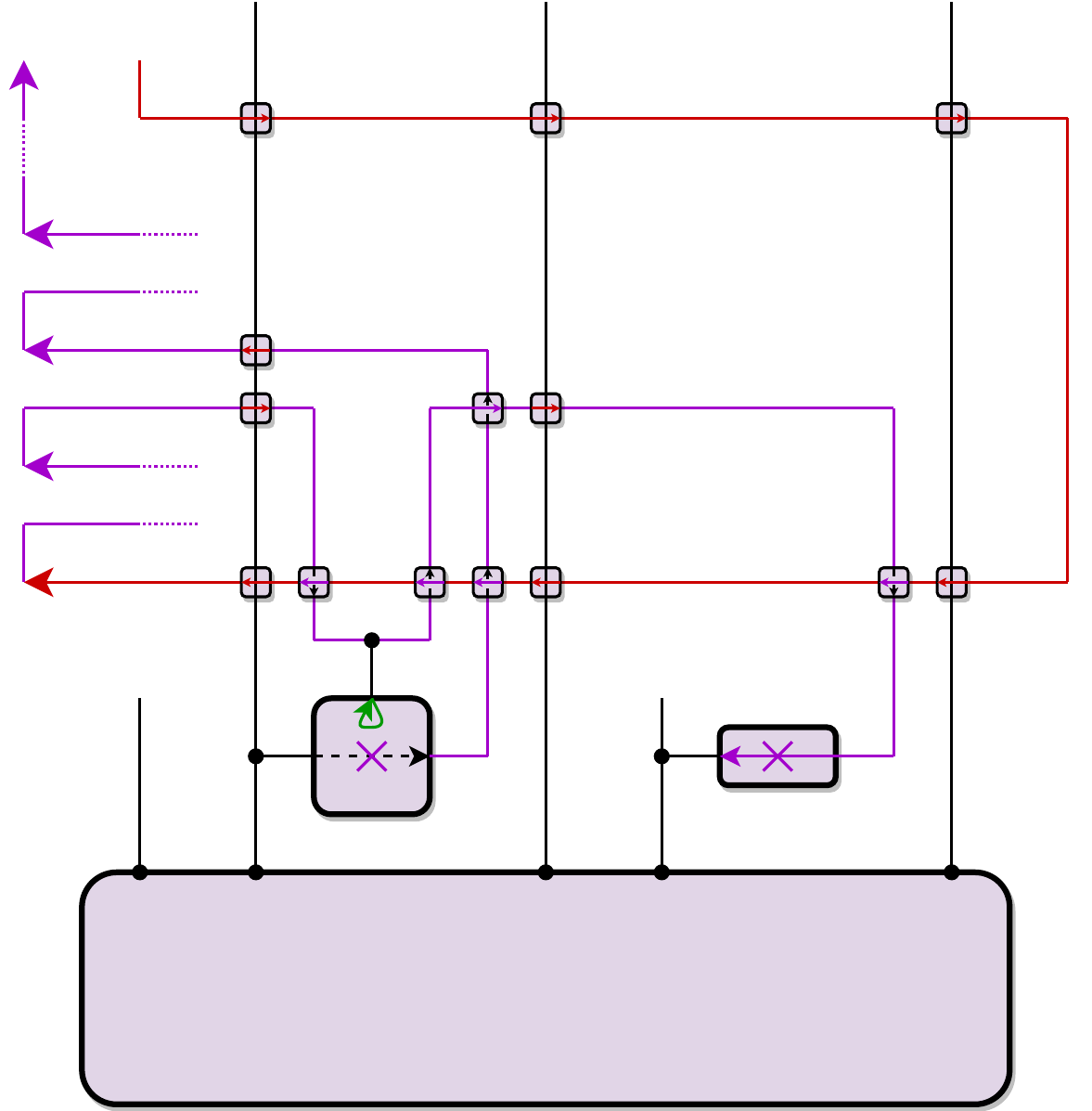}
    \put(13,98){\makebox(0,0)[c]{\strut \(\cIn\)}}
    \put(2,98){\makebox(0,0)[c]{\strut \(\cOut\)}}
    \put(23,19){\makebox(0,0)[c]{\strut \(b_i\)}}
    \put(59.5,19){\makebox(0,0)[c]{\strut \(a_i\)}}
    \put(34,24.5){\makebox(0,0)[c]{\strut \(O_i\)}}
    \put(70,26.5){\makebox(0,0)[c]{\strut \(D_i\)}}
    \put(6,54.5){\makebox(0,0)[l]{\footnotesize to $O_1, D_1$}}
    \put(6,59.75){\makebox(0,0)[l]{\footnotesize from $O_{i-1}$}}
    \put(6,65){\makebox(0,0)[l]{\footnotesize to $O_i, D_i$}}
    \put(6,70.25){\makebox(0,0)[l]{\footnotesize from $O_i$}}
    \put(6,75.5){\makebox(0,0)[l]{\footnotesize to $O_{i+1}, D_{i+1}$}}
    \put(6,80.75){\makebox(0,0)[l]{\footnotesize from $O_k$}}
    \put(49,11){\makebox(0,0)[c]{\LARGE $G$}}
  \end{overpic}
  \caption{Simulation of a simply checkable $G^C$, adapted from \cite[Figure~22]{GadgetsChecked_FUN2022} using \SCX{} gadgets instead of \WCX.}
  \label{fig:simply checkable postselect}
\end{figure}

  We must prove that, for any traversal sequence $X$ satisfying $XC \in G$,
  the traversal sequence $X[\cIn \to \cOut]$ is accepted by our simulation;
  and conversely that any traversal sequence accepted by our simulation with a nontrivial traversal involving $\cIn$ or $\cOut$ is of the form $X[\cIn \to \cOut]$ with $XC \in G$.

  The only way to reach $\cOut$ is
  by traversing $O_k$.
  But in order to traverse each $O_{i+1}$,
  it is necessary to first open it;
  and the opening entrance of $O_{i+1}$ is reachable only by traversing $O_i$.
  Therefore reaching $\cOut$ can be done only by first opening $O_1$, which requires entering at $\cIn$.

  Now suppose the agent enters the simulation at $\cIn$.
  It is then forced to close all of the edges to the outside via \SCX{} traversals
  and then perform the following actions for each $i$ in order:
  open $O_i$, traverse $D_i$, enter $G$ through $a_i$,
  exit $G$ through $b_i$, and traverse $O_i$.
  After this, it is finally forced to exit at $\cOut$.
  Note that the agent has no choice at any point in this process;
  any nontrivial deviation from the above plan results in getting stuck within the simulation.

  It follows that, if the agent makes any nontrivial traversal involving $\cIn$ or $\cOut$,
  then the overall sequence of traversals on the simulation
  is of the form $X[\cIn \to \cOut]$ where $XC \in G$,
  because the sequence of traversals on $G$ performed along the $\cIn \to \cOut$ path is just~$C$
  (or more precisely, the traversal sequence can be contracted to $C$ according to the rules in Definition~\ref{defn:gizmo}).

  On the other hand, the above sequence of actions shows that $X[\cIn \to \cOut]$ is accepted by the simulation provided $XC \in G$.
  So the simulation is a simply checkable $\pss{G}{C}$.
\end{proof}

\subsection{Leaky Checkable Gizmos Framework}
\label{sec:leaky}

For completeness and comparison's sake,
we state the leaky checkable gizmos framework obtained
by directly generalizing the checkable gadgets framework of \cite{GadgetsChecked_FUN2022} to arbitrary gizmos, without replacing
weak closing crossovers with strong ones.
The benefit of the leaky framework is that it only requires the \SO{} and \MSC{} base gadgets, without needing the stronger \SC{} gadget.
The cost is that checking may allow some leakage through exposed locations before the final check.
Thus a leaky check does not necessarily force the exact postselection $\pssl{G}{C}{L'}$;
instead it forces a gizmo containing all definitely checkable sequences and only sequences that, after possible additional moves on $L'$, can be completed to pass the check.

\begin{definition}
  Let $G$ be a gizmo, $L'$ be a subset of $\locs{G}$,
  and $C$ be a traversal sequence on $\locs{G}$.
  Gizmo $H$ is a \defn{weak postselection} of $(G, C, L')$
  if it satisfies the following conditions.
  \begin{enumerate}
  \item $\locs{H} = L'$.
  \item If $XC \in G$, then $X \in H$ for every traversal sequence $X$ on $L'$.
  \item If $X \in H$, then there exists a traversal sequence $Y$ on $L'$ so that $XYC \in G$.
  \end{enumerate}
\end{definition}

The (strong) postselection $\pssl{G}{C}{L'}$ is always a
weak postselection of $(G, C, L')$.
If $\pssl{G}{C}{L'}$ is prefix-closed (as in a reachability problem),
then it is the only one.
More generally, the weak postselections of $(G,C,L')$ are the gizmos $H$ on $L'$
satisfying $\pssl{G}{C}{L'}\subset H \subset G'$,%
\footnote{Because gizmos are sets of legal traversal sequences, $G\subset H$ means $H$ allows everything $G$ does.}
where $G'$ is the set of prefixes of traversal sequences in $\pssl{G}{C}{L'}$ (which is a gizmo).

The leaky framework uses a weaker notion of nonlocal simulation:

\begin{definition}
  A set $S$ of gizmos \defn{[planarly] leakily nonlocally simulates}
  a prefix-closed gizmo $H$ if,
  for every set of \emph{prefix-closed} gizmos $T$,
  there is a polynomial-time reduction
  from [planar] reachability on $\{H\} \cup T$
  to [planar] targeted set reconfiguration on $S \cup T$.
\end{definition}

The leaky analog to Theorem~\ref{thm:check} is as follows; it corresponds to \cite[Theorem~5.3]{GadgetsChecked_FUN2022}.

\begin{theorem}\label{thm:leaky check}
  Let $G$ be a gizmo, $L'$ be a subset of $\locs{G}$, and $C$ be a traversal sequence on $\locs{G}$,
  such that $\pssl{G}{C}{L'}$ is prefix-closed.
  Then $\{G, \SO, \MSC\}$ planarly leakily nonlocally simulates $\pssl{G}{C}{L'}$.
\end{theorem}

This theorem is shown by simulating a leakily simply checkable $H$, where $H$ is some weak postselection of $(G, C, L')$, which in fact implies $H = \pssl{G}{C}{L'}$. Here ``leakily simply checkable'' refers to the following definition, corresponding to \cite[Definition~5.5]{GadgetsChecked_FUN2022}.%
\footnote{This definition fixes a minor error in \cite{GadgetsChecked_FUN2022} by correctly accounting for all traversals involving $\cIn$ and $\cOut$.}

\begin{definition}
  Let $G, H$ be gizmos with $\locs{G} = \locs{H} \sqcup \{\cIn, \cOut\}$.
  For planar gizmos, this decomposition must respect the cyclic order.
  Then $G$ is a \defn{leakily simply checkable} \(H\)
  if it satisfies the following conditions.%
  \begin{enumerate}
  \item If \(X \in H\), then \(X[\cIn \to \cOut] \in G\).
  \item For any traversal sequence \(Y \in G\),
    let \(\widetilde Y\) be the traversal sequence obtained from \(Y\)
    by removing any instances of the trivial traversals
    \(\cIn \to \cIn\) or \(\cOut \to \cOut\).
    Then either $\widetilde Y$ contains only traversals on $\locs{H}$, or $\widetilde Y$ is of the form $X[\ell_1 \to \ell_2, \ell_3 \to \ell_4, \ldots, \ell_{n-1} \to \ell_n]$ where $X \in H$,
    $\ell_1 = \cIn$, $\ell_n = \cOut$, and the locations $\ell_2, \dots, \ell_{n-1}$ (if $n > 2$) belong to $\locs{H}$.
  \end{enumerate}
\end{definition}

\section{Graph Orientation Reconfiguration}
\label{sec:GORP}

In this section, we give some general connections between NCL, a graph orientation problem, and the gadgets framework. We use this to show that the reconfiguration problem for a certain class of gadgets is \PSPACE-complete.

\subsection{Graph Orientation}

Graph Orientation (GO) problems are a subclass of constraint satisfaction problems defined on undirected graphs. The goal is to orient every edge of the graph such that certain local constraints at each vertex are satisfied.
An example of such a problem is ``1-in-3 GO'', in which the graph is 3-regular and there are three types of vertex constraints: (1)~the indegree is exactly~1; (2)~the indegree is exactly~2; (3)~the indegree is either 0 or~3.  Horiyama et al.~\cite{horiyama2012packing,horiyama2017tiling} showed that 1-in-3 GO is \NP-complete (in fact, ASP-complete).

Now we consider more general types of vertex constraints, which may not treat all their incoming edges the same.  A \defn{vertex type} on a set $L$ of \emph{locations} consists of a subset $U \subset 2^L$, which we think of as the valid sets of incoming edges.  For instance, the second type of vertex in 1-in-3 GO would be represented on locations $a, b, c$ by the set $U = \{\{a, b\}, \{a, c\}, \{b, c\}\}$.

For a fixed set of allowed vertex types, the \defn{Graph Orientation Problem} is as follows.  An instance of the problem is an undirected multigraph without self-loops, where each vertex $v$ is labeled with a vertex type $U_v$ from the allowed set, and the edges incident to $v$ are locally labeled with the locations of $U_v$.  A solution consists of an orientation of the edges so that the set of locations corresponding to incoming edges at each vertex $v$ is an element of $U_v$.
This problem can be thought of as a Boolean constraint satisfaction problem in which every variable occurs in exactly two clauses and exactly one appearance is negated.

In the \emph{planar} setting, each vertex type includes a cyclic ordering on its locations, and instances must be planar embeddings of graphs that are locally consistent with these cyclic orderings.

\subsection{GORP and NCL}

Next we define the \defn{Graph Orientation Reconfiguration Problem (GORP)} in which we are given an instance of GO, together with initial and final orientations of the graph that are both valid solutions.  The problem is to determine whether there exists a sequence of edge flips that takes the initial orientation to the final orientation while each intermediate orientation also satisfies all vertex constraints.
\defn{Planar GORP} is defined analogously in terms of planar GO.

An important special case of GORP is \defn{Nondeterministic Constraint Logic (NCL)}.
In NCL, every edge is assigned a nonnegative weight, and every vertex is assigned a target weight.
A valid configuration is an orientation of the edges such that the sum of the weights of incoming edges at each vertex is at least the vertex's target weight.
It turns out that just two vertex types capture the complexity of NCL.
An \defn{AND vertex} is a degree-3 vertex incident to two weight-1 edges and one weight-2 edge, with a target weight of~2.
An \defn{OR vertex} is a degree-3 vertex incident to three weight-2 edges, with a target weight of~2.

\begin{theorem}[\cite{NCL_TCS,hearn2009games}]\label{thm:ncl pspace}
  Nondeterministic Constraint Logic is \PSPACE-complete, even when restricted to simple planar graphs where every vertex is either an AND vertex or an OR vertex.
\end{theorem}

\subsection{GORP Variants}

To show a connection between the agentless GORP problem and the agentful gizmo reconfiguration problem, we will actually define three variants of the GORP problem --- synchronous, asynchronous, and token ---
and show that they are equivalent under certain conditions.
Figure~\ref{fig:gorp variants} illustrates the three variants of GORP.

In the \defn{synchronous} variant described above, a single move consists of flipping the orientation of a single edge in the graph.

In the \defn{asynchronous} variant
(defined for NCL in \cite{viglietta2013partial}),
every edge is subdivided into two half-edges, each of which has its own independent orientation.
Each vertex only constrains the orientations of its incident half-edges.
A half-edge is oriented \defn{vertexwards} if it points towards its incident vertex, and \defn{edgewards} if it points away from its incident vertex.
A move consists of flipping the orientation of a single half-edge, subject to the constraint that the two halves of an edge cannot simultaneously point vertexwards.

In the \defn{token} variant of GORP,
every half-edge has its own orientation, like in asynchronous GORP,
but now we allow the two halves of an edge to simultaneously point vertexwards.
There is also a single \emph{token}, which can be either absent from the graph or located on one of the edges; the token can be thought of as an ``agent'' that travels around the graph and flips edges.
A move consists of one of the following:
\begin{enumerate}
\item Flipping any half-edge so that it points vertexwards.
\item Moving the token to any edge, provided that at least one of the halves of that edge points edgewards.
\item Flipping either half of the edge with the token.
\item Removing the token from the graph.
\end{enumerate}
The problem is to find a sequence of moves that brings the initial orientation to the final orientation, where the token is absent from the graph in both the initial and final states.
As always, the intermediate states must satisfy the vertex constraints.

\begin{figure}
  \centering
  \def\scale{2}
  \def\width{2.9cm}
  \subcaptionbox{\label{fig:gorp sync}\centering Synchronous}[\width]{
    \includegraphics[scale=\scale]{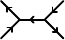} \\[.25cm]
    \includegraphics[scale=\scale]{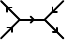}
  }
  \hspace{2cm}
  \subcaptionbox{\label{fig:gorp async}\centering Asynchronous}[\width]{
    \includegraphics[scale=\scale]{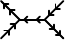} \\[.25cm]
    \includegraphics[scale=\scale]{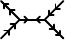} \\[.25cm]
    \includegraphics[scale=\scale]{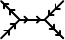}
  }
  \hspace{2cm}
  \subcaptionbox{\label{fig:gorp token}\centering Token}[\width]{
    \includegraphics[scale=\scale]{gorp/async_0} \\[.25cm]
    \includegraphics[scale=\scale]{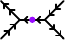} \\[.25cm]
    \includegraphics[scale=\scale]{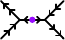} \\[.25cm]
    \includegraphics[scale=\scale]{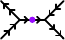} \\[.25cm]
    \includegraphics[scale=\scale]{gorp/async_2}
  }
  \caption{Flipping an edge's orientation in three different variants of GORP.}
  \label{fig:gorp variants}
\end{figure}

It turns out that all three variants of the problem are equivalent provided that every vertex type $U$ is \defn{upward-closed}, meaning that, if $S \subset S'$ and $S \in U$, then $S' \in U$ also. Informally, having more edges pointing in never causes a vertex to become unsatisfied. In particular, the vertex constraints in NCL satisfy this property.

\begin{lemma}
  \label{lem:gorp equivalence}
  Let $T = (V, E)$ be a multigraph without self-loops,
  where each $v \in V$ is labeled with an upward-closed vertex type $U_v$ and
  the edges incident to $v$ are locally labeled with the locations of $U_v$.
  Let $I$ and $F$ be initial and final orientations of $T$ that satisfy the vertices.
  Then the following are equivalent:
  \begin{enumerate}
  \item $I$ can be reconfigured to $F$ via synchronous GORP moves.
  \item $I$ can be reconfigured to $F$ via asynchronous GORP moves.
  \item $I$ can be reconfigured to $F$ via token GORP moves.
  \end{enumerate}
\end{lemma}
\begin{proof}
  We show how to convert solutions between these problems in a cycle.

  First, we show how to convert an asynchronous solution to a synchronous one.
  Each time we flip a half-edge to point vertexwards, replace that move with a synchronous move orienting the full edge towards that vertex.
  By the definition of the asynchronous problem,
  the other half-edge of that edge must be pointing edgewards.
  Thus we maintain the invariant that
  the set of edges $S_v$ pointing towards $v$ is a superset of the set of edges pointing towards $v$ in the asynchronous solution.
  Because each $U_v$ is upward-closed, $S_v$ is still an element of $U_v$, so this is still a valid move sequence.

  Next, we show how to convert a synchronous solution to a token GORP solution. For each synchronous edge flip that points the edge $e$ away from vertex $u$ and towards vertex $v$, replace it with the following sequence of token GORP moves: move the token to $e$, flip the half-edge incident to $u$ edgewards, flip the half-edge incident to $v$ vertexwards, and then remove the token from the graph.

  Finally, we show how to convert a token solution to an asynchronous GORP solution.
  We can \emph{almost} do this by simply forgetting about the token.
  The only issue is that token GORP allows orienting both halves of an edge vertexwards.
  We will show how to modify the solution to avoid doing this.

  We claim that,
  in any reconfiguration from $I$ to $F$ by token GORP moves,
  an edge cannot have both halves pointing vertexwards except while the token is located at that edge.
  Suppose for contradiction that both halves of an edge $e$ are oriented vertexwards while the token is not at~$e$.
  Then it is impossible for the token to ever revisit $e$, so the orientations of $e$'s half-edges can never be changed.
  As a consequence, the final orientation $F$ (which by definition assigns a single orientation to~$e$) is unreachable, a contradiction.

  Now suppose that the given token GORP solution at some point makes both halves of an edge $e$ point vertexwards.
  By the claim, this only occurs while the token is located at~$e$.
  Let $X$ be the subsequence of moves that occur during this visit of the token to the edge $e$.
  Let $S$ and $T$ be the states of the graph before and after $X$, respectively.
  We replace $X$ with the following sequence of moves.
  First, we make all moves in $X$ that re-orient a half-edge not in~$e$;
  by definition of token GORP moves, these are all made vertexwards.
  Next, we point edgewards any half of $e$ that is edgewards in~$T$.
  Finally, we point vertexwards any half of $e$ that is vertexwards in~$T$.
  Notably, during this visit of the token to $e$, every edgewards flip of $e$ occurs before every vertexwards flip.

  Next we show that performing this replacement preserves that we have a token GORP solution.
  Because our vertex types are upward-closed, flipping any half-edge vertexwards is always legal, so the only moves we need to check are edgewards flips of halves of $e$.
  Such moves are legal because the token is at $e$, and the vertex to which the half-edge connects is satisfied because it is now in the same state as in~$T$.
  The final state after performing this replacement is exactly $T$, because the only moves that can happen while the token is at $e$ are vertexwards flips not on $e$ (which we replicated) and arbitrary flips of $e$ (which yield the same final state as our new solution).

  We perform this replacement for every visit of the token to an edge.
  The result is a new token GORP solution that satisfies,
  during every visit of the token to an edge~$e$,
  that all edgewards flips of $e$ occur before any vertexwards flips.
  By the claim, we thus never have both halves of an edge $e$ pointing vertexwards.
  Therefore we can convert this solution to an asynchronous solution by simply forgetting about the token.
\end{proof}

Another useful fact about GORP with upward-closed vertex types is that
the problem is unaffected by subdividing any edge by a \defn{wire vertex}
--- a vertex on two locations $\{a, b\}$ whose type is $U_{\text{wire}} = \{\{a\}, \{b\}, \{a, b\}\}$:

\begin{corollary}
  \label{cor:gorp wire equivalence}
  Consider a GORP instance in which all vertices have upward-closed vertex types.
  Then any edge can be subdivided by a wire vertex,
  and the resulting instance has a solution if and only if the original did.
\end{corollary}
\begin{proof}
  The constraint on a wire vertex is exactly the same as the constraint
  on half-edges in asynchronous GORP.
  So we can convert between solutions on the original and solutions on the modified instance
  in the same way as we convert between synchronous and asynchronous solutions
  in Lemma~\ref{lem:gorp equivalence},
  but only considering a single edge.
\end{proof}

\subsection{GORP Gizmos}

In this section, we define a relationship between upward-closed GORP problems and gizmo (targeted set) reconfiguration.  We will describe sufficient conditions for a gizmo to correspond to a GORP vertex type, which will allow a reduction from the GORP problem to the corresponding gizmo reconfiguration problem.

\label{sec:open/closed intuition}
The intuition behind the reduction is as follows.
For every GORP vertex, there will be a corresponding \emph{GORP gizmo} that constrains the motion-planning agent to move in the same way as the token in token GORP.
Every GORP gizmo will have a location $x$ that the agent can freely visit at any time, and a location corresponding to each of the incident half-edges in the GORP problem.
In particular, the state $S$ of a vertex is both the set of locations whose corresponding half-edges are oriented vertexwards and the set of locations through which the agent cannot leave the gizmo.
To flip a half-edge toward a vertex, the agent will enter the vertex's gizmo through $x$, close the corresponding location (so it can no longer leave through that location), then leave through $x$.
The only way to re-open that location will be by entering the vertex from that corresponding edge.
To visit an edge, the agent will enter a vertex's gizmo through $x$ and then leave through an open location into the corresponding edge.
This approach enforces the rule that a token may move to an edge only if at least one of the edge's halves points edgewards.
To enforce the GORP constraint of the vertex, the GORP gizmo must maintain an invariant that enough locations are always closed.

Now we formalize this notion of correspondence between GORP vertex types and gizmos.
Let $U$ be an upward-closed GORP vertex type on location set $L$ with a designated initial state $I \in U$, and suppose $G$ is a gizmo whose location set is $L(G) = L \sqcup \{x\}$.
In the planar setting, we require the cyclic orderings to agree on~$L$.
The next two definitions capture the two directions needed for correspondence.
``Soundness'' says that any traversal sequence of $G$ corresponds to a valid history of moves between states of $U$.
``Completeness'' says that any such valid history of moves can be realized by a traversal sequence of $G$.

\begin{definition}
  \label{def:sound}
  The gizmo $G$ is \defn{sound} with respect to $(U, I)$ if the following conditions hold for every traversal sequence $X \in G$.
  Write $X$ as a sequence of $n$ traversals $[X_1, \dots, X_n]$.
  We require that there exists a sequence of $n+1$ states $S_0, \dots, S_n$ such that:
  \begin{enumerate}
  \item Each $S_i$ is an element of $U$.
  \item $S_0 \subset I$.
  \item If $X_{i+1} = \ell \to a$ for any $\ell \ne a$, then $a \notin S_i$.
  \item If $a \in S_i$ but $a \notin S_{i+1}$, then $X_{i+1} = a \to \ell$ for some $\ell \in \locs{G}$.
  \end{enumerate}
\end{definition}
\begin{definition}
  \label{def:complete}
  The gizmo $G$ is \defn{complete} with respect to $(U, I)$ if the following conditions hold for every sequence $S_0, \dots, S_n$ of states of $U$ starting with $S_0 = I$,
  where each state differs from the previous by a single element.
  We require that $X \in G$, where $X = X_1 X_2 \dots X_n C$ is defined as follows:
  \begin{enumerate}
  \item If $S_{i+1} = S_i \sqcup \{a\}$, then $X_{i+1} = [x \to a, x \to x]$.
  \item If $S_{i+1} = S_i \setminus \{a\}$, then $X_{i+1} = [x \to x, a \to x]$.
  \item $C$ is the concatenation of the traversals $x \to a$ for each $a \notin S_n$ (in any order), followed by $x \to x$.
  \end{enumerate}
\end{definition}

If $G$ is both sound and complete with respect to $(U, I)$, then we say $G$ \defn{corresponds} to $(U, I)$.

Intuitively, the $x \to a$ and $a \to x$ traversals in the above definitions correspond to the agent traveling along an edge between two vertices and flipping that edge. The $x \to x$ traversals are not strictly needed to get completeness for Theorem~\ref{thm:gorp reduction}, but they allow us to build gadgets that need the agent to return to $x$ to fix the internal state of the gadget after flipping an edge without breaking completeness, which is useful for our reduction to Push-1.

Figure~\ref{fig:NCL AND} shows an example of a Push-1 gizmo corresponding to the AND vertex, which is the upward-closure of $\{\{a, b\},\{c\}\}$.
In this example, each of the locations $a, b, c$ has a block in front of it (and thus no traversal can exit through that location) precisely when that location is in the state $S$.

In the planar setting, we also need a crossover gizmo to allow the agent to reach all vertices of the graph.
Intuitively, this gadget has a passage between $x$ and $y$ that is always available to the agent,
while the other two locations $a$ and $b$ behave like a GORP wire once $x$ and $y$ are identified.
The cyclic order requires these two roles to cross geometrically.
\begin{definition}\label{def:gorp crossover}
  A gizmo $H$ on the location set $\{a, b, x, y\}$ is a \defn{GORP crossover} if it satisfies the following properties:
  \begin{enumerate}
  \item For any traversal sequence $XY \in H$, we also have $X[x \to y]Y \in H$ and $X[y \to x]Y \in H$.
  \item The quotient gizmo $H/\{(x, y)\}$ (from Definition~\ref{def:quotient}) corresponds to the GORP wire vertex on $\{a,b\}$ with initial state~$\{a\}$.
  \item The pairs $\{a, b\}$ and $\{x, y\}$ are interleaved in the cyclic order.
  \end{enumerate}
\end{definition}

Now we can state the main result of this section:

\begin{theorem}
  \label{thm:gorp reduction}
  Let $\{U_i\}$ be a collection of upward-closed GORP vertex types, and let $I_i, F_i \in U_i$ specify their initial and final states respectively.
  Suppose $\{G_i\}$ is a collection of gizmos such that each $G_i$ corresponds to $(U_i, I_i)$.
  Then there is a polynomial-time reduction from GORP\/\footnote{By Lemma~\ref{lem:gorp equivalence}, the specific variant of GORP is irrelevant.} using vertices with types $U_i$ and initial/final states $I_i, F_i$
  to the gizmo targeted set reconfiguration problem with $\{\pss{G_i}{C_i}\}$, where $C_i$ is the concatenation of traversals $x \to a$ for each location $a \notin F_i$, followed by $x \to x$.

  Furthermore, if $H$ is a GORP crossover, then there is a polynomial-time reduction from planar GORP to the planar gizmo targeted set reconfiguration problem on $\{\pss{G_i}{C_i}, \allowbreak \pss{H}{[x \to a, x \to x]}, \allowbreak \pss{H}{[x \to b, x \to x]}\}$.
\end{theorem}

We begin by constructing the gizmo reconfiguration instance for the reduction in the nonplanar setting.
Suppose we are given an instance of the graph orientation reconfiguration problem.
In particular, we are given a multigraph $T = (V, E)$
where each $v \in V$ is labeled with a vertex type $U_v$,
and the edges incident to $v$ are locally labeled with the locations of $U_v$.
We are also given initial and final orientations $I, F$
that are consistent at each vertex $v$ with $I_v$ and $F_v$, respectively.

We build our instance of the gizmo reconfiguration problem as follows.
For each vertex $v$, we have a gizmo $g_v = \pss{G_v}{C_v}$.
We define the equivalence relation $\sim$ such that,
for every edge $(u, v)$ labeled with location $a$ of $u$ and location $b$ of $v$,
there is a corresponding equivalence $a \sim b$ between location $a$ of $g_u$ and location $b$ of $g_v$.
Let $x_v$ denote location $x$ of~$g_v$.
We also identify the locations $x_u \sim x_v$ for every pair of vertices $u, v$,
forming a single location (equivalence class) $x$ in our simulation.
We define both the initial and the final location to be~$x$.

To finish the nonplanar reduction, it remains to show that
a solution to the constructed gizmo reconfiguration instance exists
if and only if a solution to the GORP instance exists.
We prove each implication in its own lemma:

\begin{lemma}
  \label{lem:gorp sound}
  If there is a solution to the constructed gizmo targeted set reconfiguration instance,
  then there is a solution to the GORP instance.
\end{lemma}
\begin{proof}
  Suppose we have a solution $X = [X_1, X_2, \dots, X_{|X|}]$ to the gizmo problem.
  We will show that there exists a sequence of token GORP moves that reconfigures $I$ to $F$.
  The token will follow the path defined by the traversal sequence~$X$,
  where the special location $x$ represents the token being absent from the original graph.
  We will define a sequence $Q_0, \dots, Q_{|X|}$ of token GORP configurations, i.e., half-edge orientations plus token location.
  We will then show that, for each $0 \le t < |X|$, there is a sequence of token GORP moves that reconfigures $Q_t$ to $Q_{t+1}$ while moving the token along the traversal $X_{t+1}$.
  We will also show that there are sequences of token GORP moves that reconfigure $I$ to $Q_0$ and $Q_{|X|}$ to $F$, while leaving the token at $x$.
  Altogether, this amounts to a solution to the token GORP instance.

  We start by defining $Q_t$.
  The token in $Q_t$ is wherever the agent is after $t$ traversals;
  or absent from the graph if the agent is at location $x$.
  We specify the half-edge orientations in $Q_t$ by specifying,
  for each vertex~$v$, the set $Q_t[v]$ of incident half-edges
  pointing vertexwards.

  For each vertex $v$, let $Y_v$ be the subsequence of $X$ that consists of traversals of $g_v$.
  By definition of simulation, we have $Y_v \in g_v$,
  so $Y_v C_v \in G_v$.
  Let $n = |Y_v|$ and $m = |C_v|$.
  Because each $G_v$ corresponds to $(U_v, I_v)$ by assumption of Theorem~\ref{thm:gorp reduction}, $G_v$ is sound with respect to $(U_v, I_v)$.
  Thus there exists a corresponding sequence $S_0, \dots, S_{n+m}$ of states where $S_0 \subset I_v$, and this sequence satisfies the soundness conditions (Definition~\ref{def:sound}).
  Because $C_v$ is defined as the concatenation of traversals $x \to a$ for each $a \notin F_v$ followed by $x \to x$,
  it follows from soundness that $S_n \subset F_v$.

  We can consider the sequence $S_0, \dots, S_n$ as describing
  orientations of the half-edges incident to $v$ at the times between traversals in $Y_v$.
  At each integer time $0 \le t \le |X|$,
  define $Q_t[v] = S_i$
  where $i$ is the number of traversals of $Y_v$ that are among the first $t$ traversals of $X$.
  Then $Q_t$ satisfies the properties
  \begin{align*}
    Q_0[v] &\subset I_v, \\
    Q_{|X|}[v] &\subset F_v, \\
    Q_t[v] &\in U_v.
  \end{align*}

  Now we show that there are token GORP moves connecting the $Q_t$ in sequence.
  For each time $0 \le t < |X|$, there is a single vertex $v$ where $Q_t[v]$ may differ from $Q_{t+1}[v]$, corresponding to a single traversal $\ell_1 \to \ell_2$ of $G_v$.
  We wish to find a sequence of token GORP moves that reconfigures $Q_t[v]$ to $Q_{t+1}[v]$, and moves the token from $\ell_1$ to $\ell_2$.
  We can do this as follows.
  Let $S$ denote the state of vertex $v$ during this sequence, so that initially $S = Q_t[v]$.
  \begin{enumerate}
  \item
    First, if $Q_t[v]$ is not a subset of $Q_{t+1}[v]$, then $Q_t[v] \setminus Q_{t+1}[v] = \{\ell_1\}$ where $\ell_1 \ne x$.  Because the token is located at $\ell_1$, we can flip $\ell_1$ edgewards, removing it from $S$.
  \item If $\ell_2=x$, then we remove the token from the graph. Otherwise, if $\ell_1 \ne \ell_2$, then we move the token from $\ell_1$ to $\ell_2$.  This move is allowed because $\ell_2 \notin S$.
  \item Now we know $S \subset Q_{t+1}[v]$.  We flip every half-edge in $Q_{t+1}[v] \setminus S$ vertexwards, adding them to $S$ until it equals $Q_{t+1}[v]$.
  \end{enumerate}

  Next consider the state $I$.  We know that each $Q_0[v] \subset I_v$, so we must find a sequence of token GORP moves that flips each half-edge in $I_v \setminus Q_0[v]$ edgewards.
  We can accomplish this goal provided that we are allowed to move the token to each such edge.
  This is the case because $I$ is an orientation, so every edge already has an edgewards half-edge, and edgewards flips preserve this property.

  Finally, consider the state $Q_{|X|}$.  At time $|X|$ the token is located at $x$.
  We want a sequence of token GORP moves that reconfigures $Q_{|X|}$ to $F$ and that leaves the token at $x$.
  Because each $Q_{|X|}[v] \subset F_v$, we can just flip every half-edge in $F_v \setminus Q_{|X|}[v]$ vertexwards without moving the token at all.
\end{proof}

\begin{lemma}
  \label{lem:gorp complete}
  If there is a solution to the GORP instance,
  then there is a solution to the constructed gizmo targeted set reconfiguration instance.
\end{lemma}
\begin{proof}
  Consider the sequence $M = m_1m_2\cdots m_n$ of synchronous GORP moves that solves the GORP instance.
  We will build a corresponding traversal sequence $X = X_1X_2\cdots X_n$ that solves the gizmo problem.

  Suppose the $i$th move $m_i$
  consists of flipping an edge that initially points at location $a$ of~$u$,
  so that it points in the opposite direction toward location $b$ of~$v$.
  Then the traversal sequence corresponding to $m_i$ is
  \[X_i = [x_u \to x_u, x_v \to b, a \to x_u, x_v \to x_v].\]

  We claim that $X = X_1X_2\cdots X_n$ solves the gizmo targeted set reconfiguration instance.
  First, it is a valid traversal sequence, because each $X_i$ starts and ends with $x_u \sim x_v$, whose equivalence class is also where the agent starts and ends overall, and in the middle of $X_i$, $a \sim b$.
  Now let $Y_v$ be the subsequence of $X$ corresponding to the traversals of $g_v$.
  For each state change of $U_v$, there are two traversals in $Y_v$,
  which are exactly those
  in Definition~\ref{def:complete} corresponding to the state change.
  Because $G_v$ is complete, we have $Y_v C_v \in G_v$, which means $Y_v \in \pss{G_v}{C_v}$.
\end{proof}

Lemmas~\ref{lem:gorp sound} and~\ref{lem:gorp complete} together
prove Theorem~\ref{thm:gorp reduction} in the nonplanar setting.
Now suppose we are given a planar GORP instance~$T$.
We begin with the nonplanar construction above,
most of which has the structure of the graph of $T$, which is planar.
The only issue is the identifications $x_u \sim x_v$ for each pair of vertices $u, v$:
these may cross over the edges of the original graph.

We fix this issue by replacing each crossing between a new identification and an original-graph edge $e$ with a gizmo $h = \pss{H}{[x \to \ell, x \to x]}$, where $H$ is a GORP crossover.
We orient $h$ so that the initial orientation of $e$ points from $a$ to $b$ (because the initial state in Definition~\ref{def:gorp crossover} is $\{a\}$), and we choose $\ell \in \{a, b\}$ so that the final orientation of $e$ points towards~$\ell$.
This replacement gives us a planar gizmo reconfiguration instance, shown in Figure~\ref{fig:planar gorp}.
Now we show that this planar gizmo reconfiguration instance is equivalent to the original planar GORP instance.

\begin{figure}
  \centering
  \def\scale{2.5}
  \def\width{6cm}
  \subcaptionbox{\label{fig:planar gorp instance}\centering Planar GORP}[\width]{
    \includegraphics[scale=\scale]{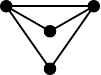}
  }
  \subcaptionbox{\label{fig:planar gorp translation}\centering Planar gizmo targeted set reconfiguration}[\width]{
    \begin{overpic}[scale=\scale]{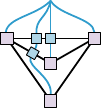}
      \put(7,64){\makebox(0,0)[c]{$g_1$}}
      \put(88,64){\makebox(0,0)[c]{$g_2$}}
      \put(47.5,41){\makebox(0,0)[c]{$g_3$}}
      \put(47.5,6){\makebox(0,0)[c]{$g_4$}}
      \put(47.5,105){\makebox(0,0)[c]{\Large $x$}}
      \put(47,65){\makebox(0,0)[c]{\footnotesize $h$}}
      \put(34,65){\makebox(0,0)[c]{\footnotesize $h$}}
      \put(30.5,51){\makebox(0,0)[c]{\footnotesize $h$}}
    \end{overpic}
  }
  \caption{Transforming an instance of planar GORP into a corresponding instance of planar gizmo reconfiguration, using GORP crossovers $h$ to connect all $x$ locations of the gizmos $g_i$.}
  \label{fig:planar gorp}
\end{figure}

\begin{lemma}
  \label{lem:gorp planar}
  There is a solution to the constructed planar gizmo reconfiguration instance
  if and only if there is a solution to the planar GORP instance.
\end{lemma}
\begin{proof}
  Let $T$ be the planar GORP instance, and let $J$ denote the planar simulation constructed above.
  First we claim that, after forgetting about planarity,
  $J$ is equivalent to $J/\sim$ where $\sim$ is the equivalence relation that identifies all $x$ and $y$ locations of all gizmos.
  This claim follows from the fact that the traversals $[x \to y]$ and $[y \to x]$ can be added to any traversal sequence of $H$,
  by Definition~\ref{def:gorp crossover}.

  Again by Definition~\ref{def:gorp crossover},
  $H/\{(x, y)\}$ corresponds to $(U_\text{wire}, \{a\})$, and we oriented the locations of each instance of $h$ in a way that agrees with the initial and final orientations of $T$.
  Thus $J/\sim$ is just the result of applying the nonplanar transformation to the GORP instance obtained by subdividing some edges of $T$ with wire vertices.
  By Corollary~\ref{cor:gorp wire equivalence}, this subdivided GORP instance is equivalent to the original GORP instance~$T$.
  By Lemmas~\ref{lem:gorp sound} and~\ref{lem:gorp complete}, $J/\sim$ has a solution if and only if there is a solution to $T$, and the same is true of $J$.
\end{proof}

\begin{proof}[Proof of Theorem~\ref{thm:gorp reduction}]
  This follows from the constructions and Lemmas~\ref{lem:gorp sound}, \ref{lem:gorp complete}, and~\ref{lem:gorp planar}.
\end{proof}

\subsection{Soundness via Closed Locations}
\label{sec:open/closed}

Before proceeding to the Push-1 proof, we give a lemma describing a sufficient condition for soundness that is easier to reason about.
Recall the intuition described in Section~\ref{sec:open/closed intuition}:
for soundness, we want to make sure that, when a location is in $S_i$,
that location is closed from within the vertex,
and the only way to re-open it is to perform a traversal that enters the vertex from that location.

To formalize this idea, for a gizmo $G$ and a traversal sequence $X$ on $G$, call a location $a$ \defn{closed for $X$} if, for every location $\ell \ne a$ and every traversal sequence $Y$, $XY[\ell\to a]\in G$ implies that $[a\to b]\in Y$ for some location $b$.
In other words, it is not possible to make a traversal that leaves through $a$ without first making a traversal that enters through $a$.
Call a location \defn{open} if it is not closed.

\begin{lemma}\label{lem:sound_s_choice}
Let $G$ be a prefix-closed gizmo, and let $R_X$ denote the set of closed locations for $X$. If $R_X\in U$ for all traversal sequences $X$, and $R_{[]} \subset I$, then $G$ is sound with respect to $(U,I)$.
\end{lemma}
\begin{proof}
We need to show that there exists a sequence $S_i$ satisfying Definition~\ref{def:sound}.  We will show that $S_i=R_{X_{:i}}$ works, where $X_{:i}$ denotes the first $i$ traversals of $X$. We need to check four conditions.

First, $S_0 = R_{[]} \subset I$. Second, $S_i \in U$ by assumption. Third, if $X_{i+1} = \ell \to a$, then $a$ cannot be closed for $X_{:i}$ because by prefix-closure $X_{:i+1} = X_{:i}[\ell \to a] \in G$, and so in the definition of closed, $Y$ would be empty. Thus $a \notin R_{X_{:i}} = S_i$. Fourth, suppose $a \in S_i$ but $a \notin S_{i+1}$. Then $a$ is closed for $X_{:i}$ and open for $X_{:i+1}$. This means that there is a traversal sequence $Y$ such that $X_{:i+1}Y[\ell\to a] \in G$ and $Y$ does not contain a traversal of the form $a \to b$ for some $b$ (because $a$ was open for $X_{:i+1}$), but also that such a traversal \emph{does} appear in $X_{i+1}Y$ (because $a$ was closed for $X_{:i}$). This means that $X_{i+1}$ must be of the form $a \to b$, satisfying the condition.
\end{proof}

\section{Push-1 Gadgets}
\label{sec:Push-1}
In this section, we show how to construct gadgets in Push-1
that correspond to the GORP vertices in Nondeterministic Constraint Logic (NCL).
For these constructions, we implement a checkable gizmo that, with the appropriate checking sequence, corresponds to the appropriate GORP vertex. %
These constructions use this extra checking sequence to maintain ``normal operation''.

We first describe the AND and OR gadgets, which are very similar.
In each construction, only three blocks can ever move
(except in the checking sequence). We will refer to these three blocks as the ``top'', ``middle'', and ``bottom'' blocks in the constructions below. The following two invariants about their positions are required for \defn{normal operation}:
\begin{enumerate}
\item Each of the three movable blocks is always on its ``track'', shown in Figures~\ref{fig:NCL AND} and~\ref{fig:NCL OR}.
\item The top and middle blocks are not touching.
\end{enumerate}
If any of these invariants are violated by the player, we say the gadget is \emph{broken}.

Our gadgets are designed so that
a location is open (in the sense of Section~\ref{sec:open/closed})
if and only if
it is directly accessible via a path from $x$ without pushing any blocks.

\subsection{AND Gadget}

\begin{figure}
\centering
\subcaptionbox{\label{fig:NCL AND1} $a$ open}{\begin{overpic}[scale=0.4]{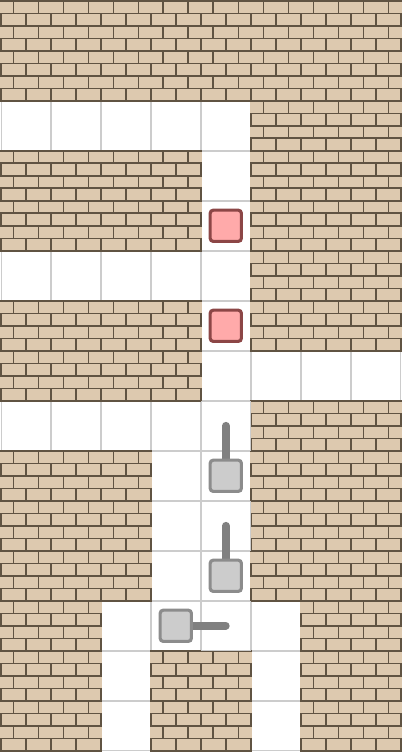}
\put(-1,43.5){\makebox(0,0)[r]{$x$}}
\put(17,-4){\makebox(0,0)[c]{$b$}}
\put(37,-4){\makebox(0,0)[c]{$c$}}
\put(54,50){\makebox(0,0)[l]{$a$}}
\end{overpic}
\vspace{0.75em}}
\hfil
\subcaptionbox{\label{fig:NCL AND2} $a$ and $b$ open}{\begin{overpic}[scale=0.4]{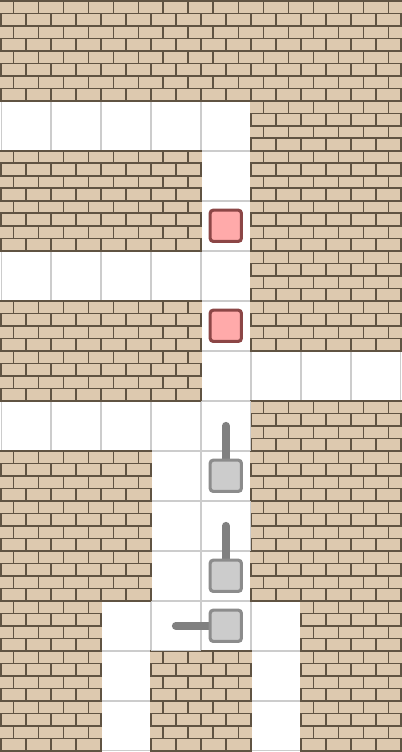}
\put(-1,43.5){\makebox(0,0)[r]{$x$}}
\put(17,-4){\makebox(0,0)[c]{$b$}}
\put(37,-4){\makebox(0,0)[c]{$c$}}
\put(54,50){\makebox(0,0)[l]{$a$}}
\end{overpic}
\vspace{0.75em}}
\hfil
\subcaptionbox{\label{fig:NCL AND4} $b$ open}{\begin{overpic}[scale=0.4]{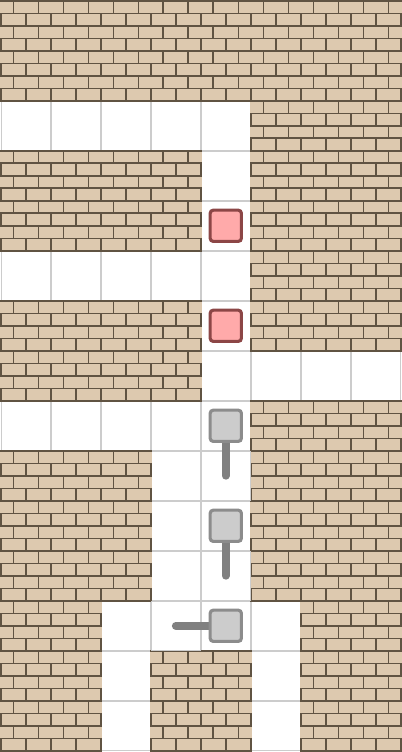}
\put(-1,43.5){\makebox(0,0)[r]{$x$}}
\put(17,-4){\makebox(0,0)[c]{$b$}}
\put(37,-4){\makebox(0,0)[c]{$c$}}
\put(54,50){\makebox(0,0)[l]{$a$}}
\end{overpic}
\vspace{0.75em}}
\hfil
\subcaptionbox{\label{fig:NCL AND3} $c$ open}{\begin{overpic}[scale=0.4]{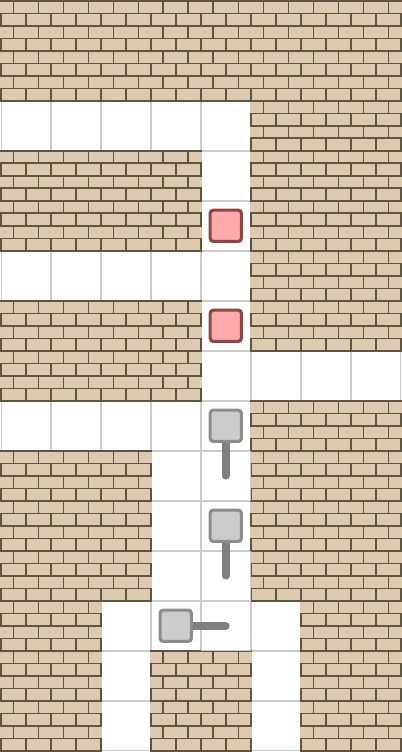}
\put(-1,43.5){\makebox(0,0)[r]{$x$}}
\put(17,-4){\makebox(0,0)[c]{$b$}}
\put(37,-4){\makebox(0,0)[c]{$c$}}
\put(54,50){\makebox(0,0)[l]{$a$}}
\end{overpic}
\vspace{0.75em}}
\caption{Checkable GORP NCL AND vertex built in Push-1.}
\label{fig:NCL AND}
\end{figure}

First, we implement a checkable GORP AND gizmo (that is, a gizmo $G$ and checking sequence $C$ such that $\pss{G}{C}$ corresponds to the GORP AND vertex) in Push-1.
Figure~\ref{fig:NCL AND} shows the gadget in the four possible states
under normal operation,
and Figure~\ref{fig:NCL AND check} illustrates the checking sequence.
Recall that the GORP AND vertex is defined as the upward closure of $\{\{a,b\},\{c\}\}$.
For a particular initial state $I$, the gadget starts in one of the configurations shown in Figures~\ref{fig:NCL AND2} and~\ref{fig:NCL AND3} such that every location not in $I$ is open.
The top two (unlabeled) locations and the red blocks are for the final checking sequence,
to be described later.

First, we will show soundness: every traversal sequence of our Push-1 gadget in normal operation corresponds to a valid sequence of GORP AND vertex states.
\begin{lemma}\label{lem:AND_sound}
In normal operation, the Push-1 AND gadget is sound with respect to the GORP AND vertex with initial state $I$.
\end{lemma}
\begin{proof}
By Lemma~\ref{lem:sound_s_choice}, it suffices to show that it is not possible for $c$ to be open simultaneously with either $a$ or $b$.
Suppose $c$ is open, meaning we can traverse from some other location to $c$ without first entering through~$c$.
Then the bottom block must be leftmost on its track, and then the middle block must be up on its track, because otherwise $c$ would be blocked and require a traversal entering through $c$ to open a path to $c$.
By definition of normal operation, the top block must also be up on its track, to avoid touching the middle block.
Thus, assuming the gadget is in normal operation, it must look like Figure~\ref{fig:NCL AND3} at this point.
In this state, however, it is clearly impossible for the agent to reach location $a$ or $b$ from any other location, so neither $a$ nor $b$ can be open while $c$ is.
\end{proof}

Next, we will show completeness: given a sequence of valid vertex states, we will construct a sequence of Push-1 moves to traverse the gadget while maintaining normal operation.
\begin{lemma}\label{lem:AND_complete}
In normal operation, the Push-1 AND gadget is complete with respect to the GORP AND vertex with initial state $I$.
\end{lemma}
\begin{proof}
We must show that every sequence $S_i$ of valid GORP AND vertex states corresponds to a sequence of traversals of the Push-1 AND gadget in normal operation.

We will construct the desired traversal sequence while maintaining the additional invariant that every location $\notin S_i$ is open after the corresponding Push-1 moves (as ultimately needed to perform~$C$).
The invariant is true initially by our choice of initial configuration, as seen in Figure~\ref{fig:NCL AND}.
It suffices to show that each state change from $S_i$ to $S_{i+1}$ preserves the invariant while performing the corresponding traversal sequence.

For each state change from $S_i$ to $S_{i+1}$, there are two cases: for some location $\ell \in \{a,b,c\}$, either $S_{i+1} = S_i \sqcup \{\ell\}$, or $S_{i+1} = S_i \setminus \{\ell\}$.
In the first case, $\ell \notin S_i$, so we do not need to push any blocks around, because $\ell$ is open and thus we can just perform $x \to \ell$ (and then $x\to x$ trivially), preserving the invariant.
In the second case, we need to make $\ell$ open while performing $[x\to x,\ell\to x]$. There are three subcases for~$\ell$.

\textbf{Subcase 1:} $\ell = a$. First, we enter at $x$, push the middle block downwards if it is not already (possible because the top and middle blocks are never adjacent), and exit at $x$. Then we enter at $a$, pushing the top block downwards if necessary, and then exit at $x$. This sequence does not change whether $b$ is open (this is determined entirely by the position of the bottom block), but necessarily closes $c$. This closing is allowed because $a \notin S_{i+1}$ implies $c \in S_{i+1}$ in order for $S_{i+1} \in U$. The resulting state will look like one of Figures~\ref{fig:NCL AND1} or~\ref{fig:NCL AND2} depending on whether $b$ was open.

\textbf{Subcase 2:} $\ell = b$. We do nothing during the $x \to x$ traversal. Then we enter at $b$, pushing the bottom block rightwards if necessary, and exit at $x$. This sequence does not affect the path to $a$, but does block $c$ (if it was not already blocked). Again this closing is allowed because $b \notin S_{i+1}$ implies $c \in S_{i+1}$ in order for $S_{i+1} \in U$. The resulting state will look like one of Figures~\ref{fig:NCL AND2} or~\ref{fig:NCL AND4} depending on whether $a$ was open.

\textbf{Subcase 3:} $\ell = c$. First, we enter at $x$, push the top block upwards if it is not already (possible because the top and middle blocks are never adjacent), and exit at $x$. Then we enter at $c$, pushing the bottom block leftwards and the middle block upwards if necessary. The final result is in Figure~\ref{fig:NCL AND3}. This sequence closes both $a$ and $b$, which is allowed because $c \notin S_{i+1}$ implies $a, b \in S_{i+1}$ in order for $S_{i+1} \in U$.

Thus the invariant is preserved under every operation, and after all state changes, all locations needed for the final sequence $C$ are open.
\end{proof}

Now we use the checkable gadgets framework from Section~\ref{sec:checkable-gadgets} to prevent the agent from deviating from normal operation.

\begin{lemma}\label{lem:AND_check}
Let $G$ be the unrestricted gizmo implemented by the Push-1 AND gadget, and let
$C=[b \to x, c \to x, a \to x, d_1 \to x, d_2 \to x]$.
Then $\pss{G}{C}$ is exactly $G$ restricted to normal operation.
\end{lemma}
\begin{proof}
The checking sequence $C$ works as follows.
First, $b \to x$ is impossible if the bottom block is not leftwards off its track, so checking $b \to x$ prevents this violation, and it also forces the bottom block rightwards.
(In normal operation, it forces the gadget into the state in either Figure~\ref{fig:NCL AND2} or~\ref{fig:NCL AND4}.)
Then checking $c \to x$ ensures that the bottom block was not pushed rightwards off its track, and that the middle block was not pushed downwards off its track.
Assuming the top and middle blocks are not adjacent (as in normal operation),
this traversal also allows the agent to push the bottom block all the way left on its track and push the middle block downwards,
as in Figure~\ref{fig:NCL ANDcheck0}.
Next checking $a \to x$ ensures that the top block was not pushed upwards off its track.
Again assuming the top and middle blocks were not adjacent
and that the agent pushed the middle block downwards in the previous step,
this traversal also allows the agent to push the top and middle blocks all the way down (off their tracks), as in Figure~\ref{fig:NCL ANDcheck1}.
Moving these blocks all the way down is possible only if the top and middle blocks were never adjacent.
Finally, checking $d_1 \to x$ and $d_2 \to x$ in order
requires pushing the red blocks all the way downwards one at a time,
as in Figures~\ref{fig:NCL ANDcheck2} and~\ref{fig:NCL ANDcheck3}.
These pushes are only possible if, in the previous steps, the agent pushed the top and middle blocks all the way downwards, thereby enforcing that these blocks were never adjacent.

\begin{figure}
\centering
\def\scale{0.38}
\subcaptionbox{\label{fig:NCL ANDcheck0} Setting up the gadget during $c\to x$ traversal.}
[0.19\linewidth]
{~\begin{overpic}[scale=\scale]{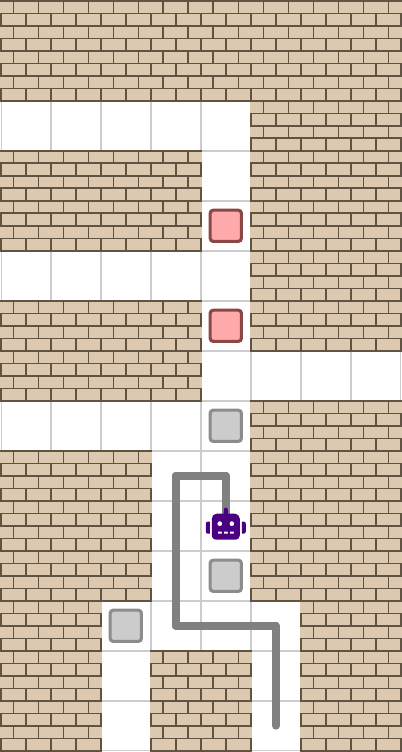}
\put(-1,43.5){\makebox(0,0)[r]{$x$}}
\put(17,-4){\makebox(0,0)[c]{$b$}}
\put(37,-4){\makebox(0,0)[c]{$c$}}
\put(54,50){\makebox(0,0)[l]{$a$}}
\put(-1,63){\makebox(0,0)[r]{$d_1$}}
\put(-1,83){\makebox(0,0)[r]{$d_2$}}
\end{overpic}
\vspace{0.75em}}%
\hfill
\subcaptionbox{\label{fig:NCL ANDcheck1} Pushing the top and middle blocks all the way down during $a \to x$ traversal.}
[0.37\linewidth]
{~\begin{overpic}[scale=\scale]{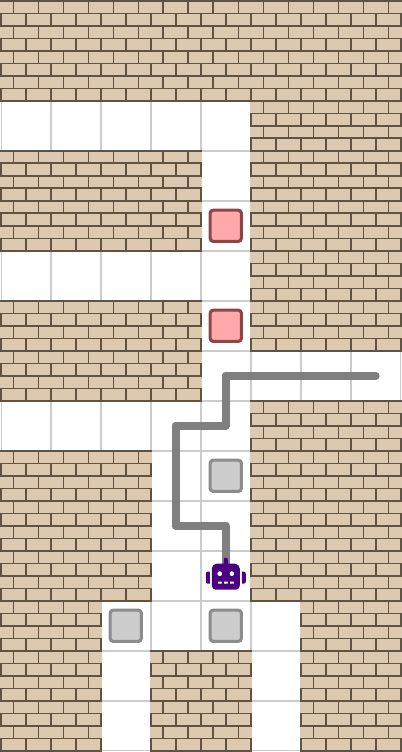}
\put(-1,43.5){\makebox(0,0)[r]{$x$}}
\put(17,-4){\makebox(0,0)[c]{$b$}}
\put(37,-4){\makebox(0,0)[c]{$c$}}
\put(54,50){\makebox(0,0)[l]{$a$}}
\put(-1,63){\makebox(0,0)[r]{$d_1$}}
\put(-1,83){\makebox(0,0)[r]{$d_2$}}
\end{overpic}~~~~~~%
\begin{overpic}[scale=\scale]{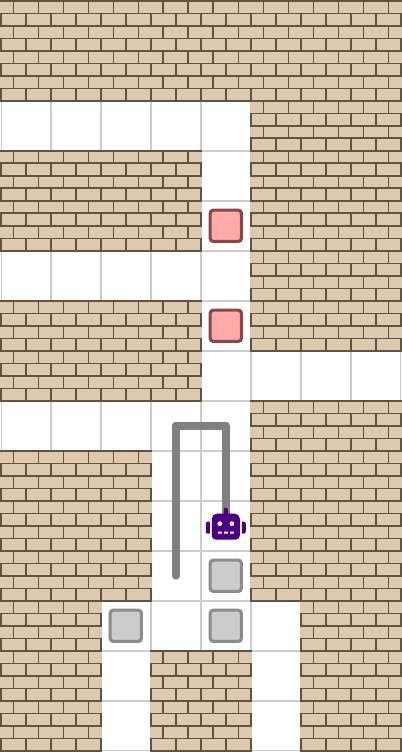}
\put(-1,43.5){\makebox(0,0)[r]{$x$}}
\put(17,-4){\makebox(0,0)[c]{$b$}}
\put(37,-4){\makebox(0,0)[c]{$c$}}
\put(54,50){\makebox(0,0)[l]{$a$}}
\put(-1,63){\makebox(0,0)[r]{$d_1$}}
\put(-1,83){\makebox(0,0)[r]{$d_2$}}
\end{overpic}
\vspace{0.75em}}%
\hfill
\subcaptionbox{\label{fig:NCL ANDcheck2} Pushing first red block down during $d_1\to x$ traversal.}
[0.19\linewidth]
{~\begin{overpic}[scale=\scale]{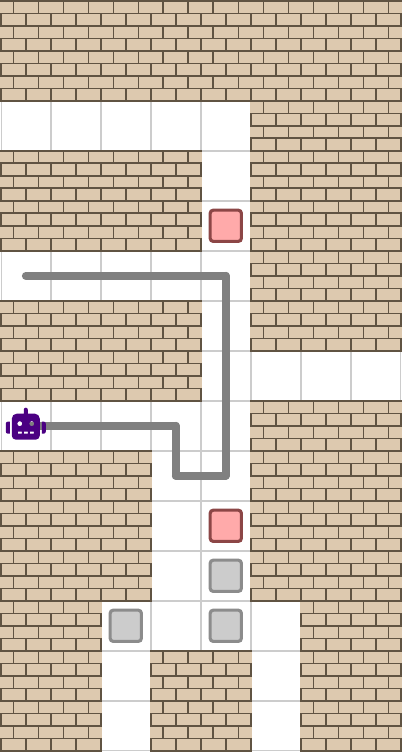}
\put(-1,43.5){\makebox(0,0)[r]{$x$}}
\put(17,-4){\makebox(0,0)[c]{$b$}}
\put(37,-4){\makebox(0,0)[c]{$c$}}
\put(54,50){\makebox(0,0)[l]{$a$}}
\put(-1,63){\makebox(0,0)[r]{$d_1$}}
\put(-1,83){\makebox(0,0)[r]{$d_2$}}
\end{overpic}
\vspace{0.75em}}%
\hfill
\subcaptionbox{\label{fig:NCL ANDcheck3} Pushing second red block down during $d_2\to x$ traversal.}
[0.19\linewidth]
{~\begin{overpic}[scale=\scale]{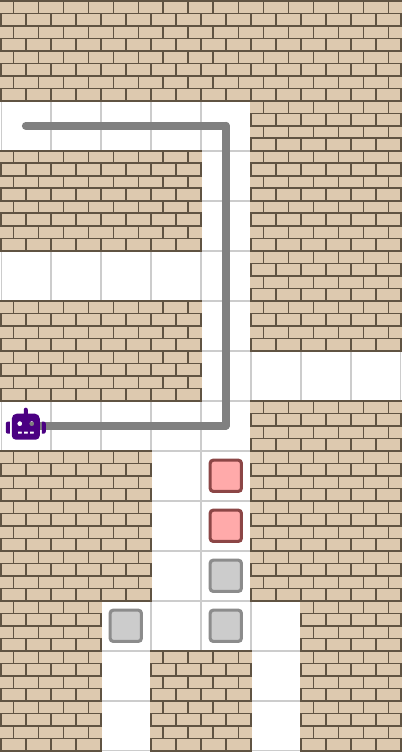}
\put(-1,43.5){\makebox(0,0)[r]{$x$}}
\put(17,-4){\makebox(0,0)[c]{$b$}}
\put(37,-4){\makebox(0,0)[c]{$c$}}
\put(54,50){\makebox(0,0)[l]{$a$}}
\put(-1,63){\makebox(0,0)[r]{$d_1$}}
\put(-1,83){\makebox(0,0)[r]{$d_2$}}
\end{overpic}
\vspace{0.75em}}%
\caption{The final checking sequence for the AND gadget.}
\label{fig:NCL AND check}
\end{figure}

The result is that the checking sequence $C$ can only be completed if the gadget was in normal operation before checking.
Conversely, the sequence of pushes described above completes $C$ from any normal-operation state.
Thus $\pss{G}{C}$ is exactly $G$ restricted to normal operation.
\end{proof}

From Lemmas~\ref{lem:AND_sound}, \ref{lem:AND_complete}, and~\ref{lem:AND_check} it follows that $G$ is a checkable GORP AND gizmo.

\subsection{OR Gadget}

\begin{figure}
\centering
\subcaptionbox{\label{fig:NCL ORa} $b$ and $c$ open}{\begin{overpic}[scale=0.4]{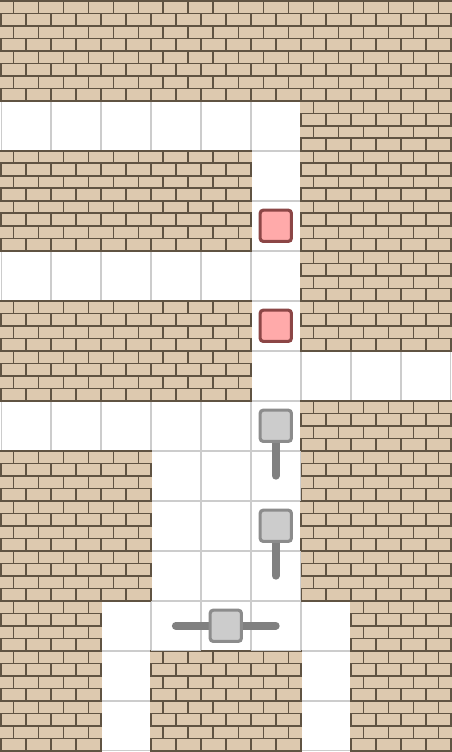}
\put(-1,43.5){\makebox(0,0)[r]{$x$}}
\put(17,-4){\makebox(0,0)[c]{$b$}}
\put(43.25,-4){\makebox(0,0)[c]{$c$}}
\put(61,50){\makebox(0,0)[l]{$a$}}
\end{overpic}
\vspace{0.75em}}
\hfil
\subcaptionbox{\label{fig:NCL ORb} $a$ and $c$ open}{\begin{overpic}[scale=0.4]{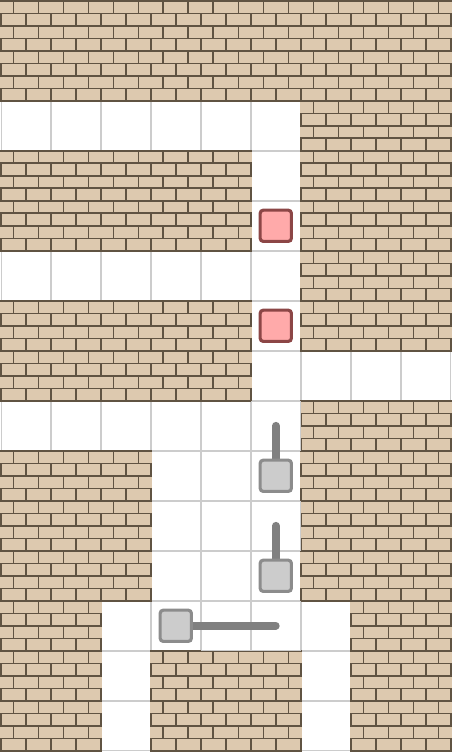}
\put(-1,43.5){\makebox(0,0)[r]{$x$}}
\put(17,-4){\makebox(0,0)[c]{$b$}}
\put(43.25,-4){\makebox(0,0)[c]{$c$}}
\put(61,50){\makebox(0,0)[l]{$a$}}
\end{overpic}
\vspace{0.75em}}
\hfil
\subcaptionbox{\label{fig:NCL ORc} $a$ and $b$ open}{\begin{overpic}[scale=0.4]{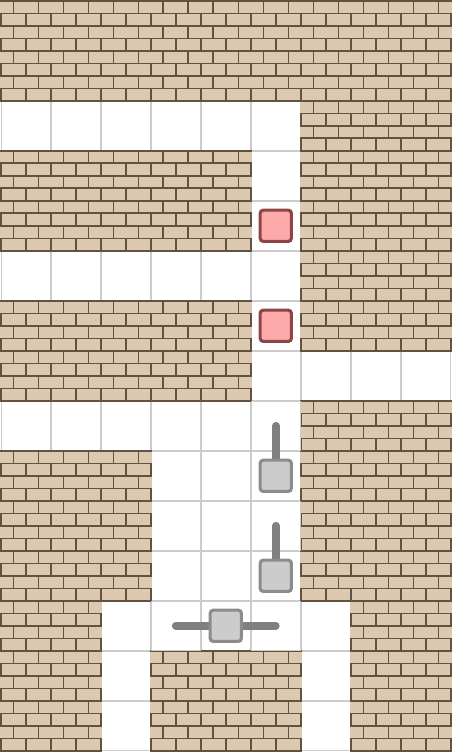}
\put(-1,43.5){\makebox(0,0)[r]{$x$}}
\put(17,-4){\makebox(0,0)[c]{$b$}}
\put(43.25,-4){\makebox(0,0)[c]{$c$}}
\put(61,50){\makebox(0,0)[l]{$a$}}
\end{overpic}
\vspace{0.75em}}
\caption{Checkable GORP NCL OR vertex built in Push-1.}
\label{fig:NCL OR}
\end{figure}

Next we implement a checkable GORP OR gizmo in Push-1.
Figure~\ref{fig:NCL OR} shows the gadget in three important states
and Figure~\ref{fig:NCL OR check} illustrates the checking sequence.
Recall that the GORP OR vertex is defined as the upward closure of $\{\{a\},\{b\},\{c\}\}$, i.e., it requires at least one of the incident edges to be oriented inward.
For a particular initial state $I$, our Push-1 gadget will start in one of the configurations in Figure~\ref{fig:NCL OR} such that every location not in $I$ is open.
The bottom block has a track of length three: there is no reason to ever put it in the rightmost position, but we also do not prevent the player from doing so.

\begin{lemma}\label{lem:OR_sound}
In normal operation, the Push-1 OR gadget is sound with respect to the GORP OR vertex with initial state $I$.
\end{lemma}
\begin{proof}
By Lemma~\ref{lem:sound_s_choice}, it suffices to show that it is not possible for all three of $a$, $b$, and $c$ to be open simultaneously.
Suppose $a$ is open, meaning we can traverse from some other location to $a$ without first entering through~$a$.
Then the top block must be in the lower position on its track,
so by the definition of normal operation,
the middle block must be in the upper position on its track,
as in Figures~\ref{fig:NCL ORb} and~\ref{fig:NCL ORc}.
Similarly, if $b$ is also open, then the bottom block cannot be on the left side of its track, as in Figure~\ref{fig:NCL ORc} (or further right).
At this point, the middle and bottom blocks collectively prevent access to $c$ without first entering through $c$.
Thus $c$ cannot also be open, so all three locations cannot be open simultaneously.
\end{proof}

Next, we will show completeness: given a sequence of valid vertex states, we will construct a sequence of Push-1 moves to traverse the gadget while maintaining normal operation.

\begin{lemma}\label{lem:OR_complete}
In normal operation, the Push-1 OR gadget is complete with respect to the GORP OR vertex with initial state $I$.
\end{lemma}
\begin{proof}
The structure of the argument is the same as for AND:
for every sequence $S_i$ of valid GORP OR vertex states,
we construct a corresponding traversal sequence of the Push-1 OR gadget
in normal operation,
while maintaining the invariant that every location $\notin S_i$ is open after the corresponding Push-1 moves.
The invariant is true initially by our choice of initial configuration, as seen in Figure~\ref{fig:NCL OR}.
It remains to show that each state change from $S_i$ to $S_{i+1}$ preserves the invariant while performing the corresponding traversal sequence.

Again, for each state change from $S_i$ to $S_{i+1}$, there are two cases: for some location $\ell \in \{a,b,c\}$, either $S_{i+1} = S_i \sqcup \{\ell\}$, or $S_{i+1} = S_i \setminus \{\ell\}$.
The first case is identical to the AND proof: we do not need to push any blocks around because $\ell \notin S_i$ is open, so we can just perform $x \to \ell$ (and then $x\to x$ trivially), preserving the invariant.
In the second case, we need to make $\ell$ open while performing $[x\to x,\ell\to x]$. There are three subcases for~$\ell$.

\textbf{Subcase 1:} $\ell = a$. First, we enter at $x$.
If $c \notin S_i$, then we push the bottom block to the far left.
This move closes $b$, which is allowed because $a, c \notin S_{i+1}$ implies $b \in S_{i+1}$ in order for $S_{i+1} \in U$.
Then we push the middle block downwards if necessary (possible because the top and middle blocks are never adjacent), and exit at $x$.
Next we enter at $a$, pushing the top block downwards if necessary, and exit at $x$.
This sequence
closes $c$ only if $c \in S_i$, which is allowed.
The end result is either Figure~\ref{fig:NCL ORb} or~\ref{fig:NCL ORc}, depending on which of $b$ or $c$ is $\in S_i$.

\textbf{Subcase 2:} $\ell = b$.
While performing $x \to x$, if $c \notin S_i$,
then we push the top block and then middle block upwards if necessary.
This move closes $a$, which is allowed because
$b, c \notin S_{i+1}$ implies $a \in S_{i+1}$ in order for $S_{i+1} \in U$.
Then we enter at $b$, pushing the bottom block to the middle if necessary, and exit at $x$.
This sequence closes $c$ only if $c \in S_i$, which is allowed.
The end result is either Figure~\ref{fig:NCL ORa} or~\ref{fig:NCL ORc}, depending on which of $a$ or $c$ is $\in S_i$.

\textbf{Subcase 3:} $\ell = c$.
If $a \in S_i$, then during $x\to x$ we push the top block upwards if necessary, closing $a$.
Then we enter at $c$.
If $a \in S_i$, then we push the middle block upwards, closing $a$, and then exit at $x$, leaving the gadget in the state in Figure~\ref{fig:NCL ORa}.
If $a \notin S_i$, then we push the bottom block to the left, and then exit at $x$, leaving the gadget in the state in Figure~\ref{fig:NCL ORb} with $b$ closed, which is allowed because $a, c \notin S_{i+1}$ implies $b \in S_{i+1}$ in order for $S_{i+1} \in U$.
Thus the invariant is preserved under every operation, and after all state changes, all locations needed for the final sequence $C$ are open.
\end{proof}

Now we use the checkable gadgets framework from Section~\ref{sec:checkable-gadgets} to force the agent to follow normal operation.
The checking sequence $C$ is in fact the same as for the AND gadget.

\begin{lemma}\label{lem:OR_check}
Let $G$ be the unrestricted gizmo implemented by the Push-1 OR gadget, and let
$C=[b \to x, c \to x, a \to x, d_1 \to x, d_2 \to x]$.
Then $\pss{G}{C}$ is exactly $G$ restricted to normal operation.
\end{lemma}
\begin{proof}
The checking sequence works as follows.
First, checking $b \to x$ ensures that the bottom block was not pushed leftwards off its track.
Then checking $c \to x$ ensures that the bottom block was not pushed rightwards off its track, and that the middle block was not pushed downwards off its track (or else it would have become stuck during $b \to x$ or earlier).
Assuming the top and middle blocks are not adjacent (as in normal operation),
this traversal also allows the agent to push the bottom block all the way left on its track and push the middle block downwards on its track,
as in Figure~\ref{fig:NCL ORcheck0}.
Next checking $a \to x$ ensures that the top block was not pushed upwards off its track.
Again assuming the top and middle blocks were not adjacent
and that the agent pushed the middle block downwards in the previous step,
this traversal also allows the agent to push the top and middle blocks all the way down (off their tracks), as in Figure~\ref{fig:NCL ORcheck1}.
Moving these blocks all the way down is possible only if the top and middle blocks were never adjacent.
Finally, checking $d_1 \to x$ and $d_2 \to x$ in order
requires pushing the red blocks all the way downwards one at a time,
as in Figures~\ref{fig:NCL ORcheck2} and~\ref{fig:NCL ORcheck3}.
These pushes are only possible if, in the previous steps, the agent pushed the top and middle blocks all the way downwards, thereby enforcing that these blocks were never adjacent.

\begin{figure}
\centering
\def\scale{0.35}
\subcaptionbox{\label{fig:NCL ORcheck0} Setting up the gadget during $c\to x$ traversal.}
[0.19\linewidth]
{~\begin{overpic}[scale=\scale]{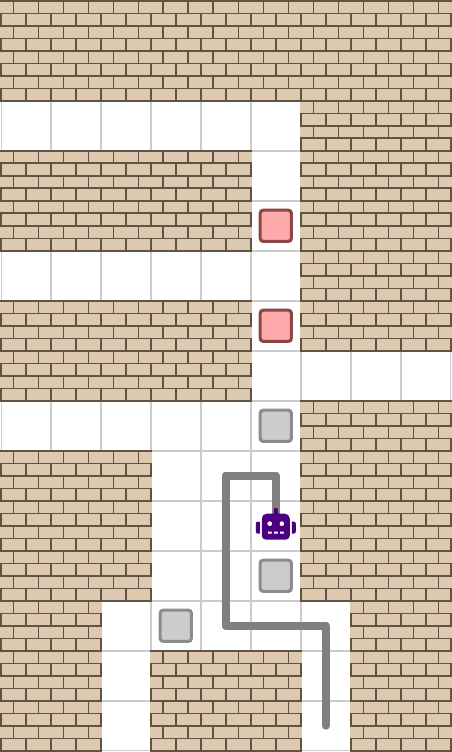}
\put(-1,43.5){\makebox(0,0)[r]{$x$}}
\put(17,-4){\makebox(0,0)[c]{$b$}}
\put(43.25,-4){\makebox(0,0)[c]{$c$}}
\put(61,50){\makebox(0,0)[l]{$a$}}
\put(-1,63){\makebox(0,0)[r]{$d_1$}}
\put(-1,83){\makebox(0,0)[r]{$d_2$}}
\end{overpic}
\vspace{0.75em}}%
\hfill
\subcaptionbox{\label{fig:NCL ORcheck1} Pushing the top and middle blocks all the way down during $a\to x$ traversal.}
[0.365\linewidth]
{~\begin{overpic}[scale=\scale]{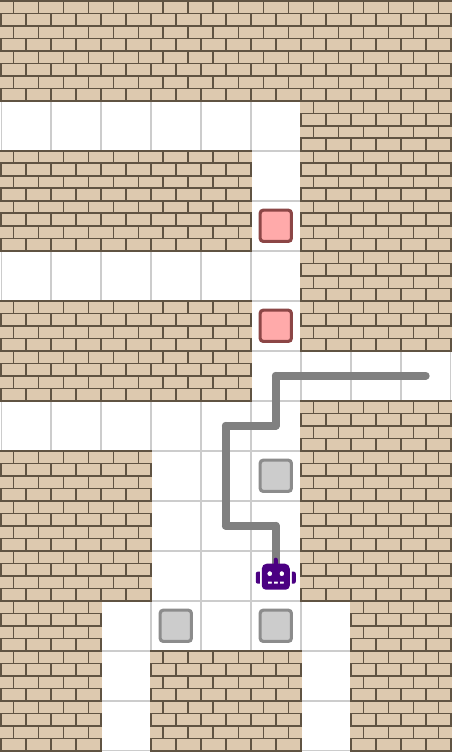}
\put(-1,43.5){\makebox(0,0)[r]{$x$}}
\put(17,-4){\makebox(0,0)[c]{$b$}}
\put(43.25,-4){\makebox(0,0)[c]{$c$}}
\put(61,50){\makebox(0,0)[l]{$a$}}
\put(-1,63){\makebox(0,0)[r]{$d_1$}}
\put(-1,83){\makebox(0,0)[r]{$d_2$}}
\end{overpic}~~~~~~~~%
\begin{overpic}[scale=\scale]{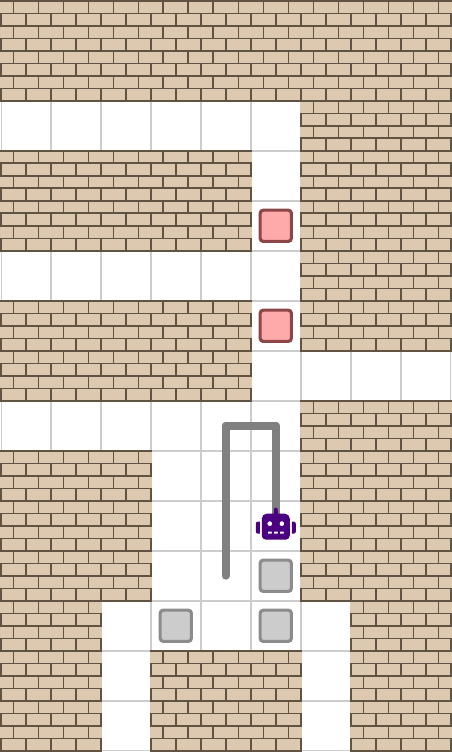}
\put(-1,43.5){\makebox(0,0)[r]{$x$}}
\put(17,-4){\makebox(0,0)[c]{$b$}}
\put(43.25,-4){\makebox(0,0)[c]{$c$}}
\put(61,50){\makebox(0,0)[l]{$a$}}
\put(-1,63){\makebox(0,0)[r]{$d_1$}}
\put(-1,83){\makebox(0,0)[r]{$d_2$}}
\end{overpic}%
\vspace{0.75em}}%
\hfill
\subcaptionbox{\label{fig:NCL ORcheck2} Pushing first red block down during $d_1\to x$ traversal.}
[0.19\linewidth]
{~\begin{overpic}[scale=\scale]{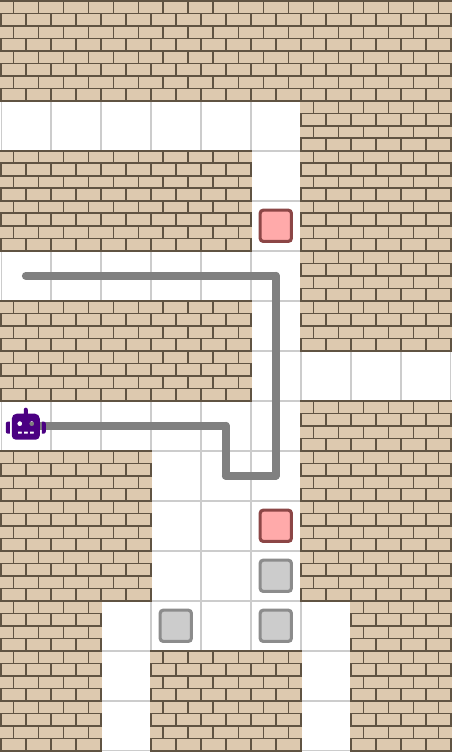}
\put(-1,43.5){\makebox(0,0)[r]{$x$}}
\put(17,-4){\makebox(0,0)[c]{$b$}}
\put(43.25,-4){\makebox(0,0)[c]{$c$}}
\put(61,50){\makebox(0,0)[l]{$a$}}
\put(-1,63){\makebox(0,0)[r]{$d_1$}}
\put(-1,83){\makebox(0,0)[r]{$d_2$}}
\end{overpic}
\vspace{0.75em}}%
\hfill
\subcaptionbox{\label{fig:NCL ORcheck3} Pushing second red block down during $d_2\to x$ traversal.}
[0.19\linewidth]
{~\begin{overpic}[scale=\scale]{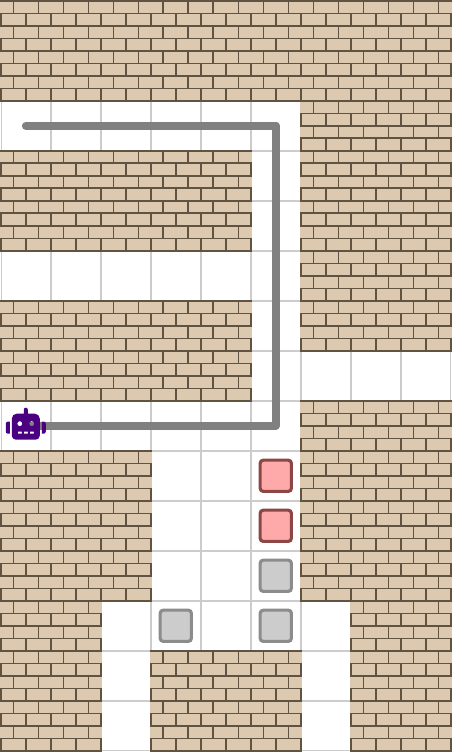}
\put(-1,43.5){\makebox(0,0)[r]{$x$}}
\put(17,-4){\makebox(0,0)[c]{$b$}}
\put(43.25,-4){\makebox(0,0)[c]{$c$}}
\put(61,50){\makebox(0,0)[l]{$a$}}
\put(-1,63){\makebox(0,0)[r]{$d_1$}}
\put(-1,83){\makebox(0,0)[r]{$d_2$}}
\end{overpic}
\vspace{0.75em}}%
\caption{The final checking sequence for the OR gadget.}
\label{fig:NCL OR check}
\end{figure}

The result is that the checking sequence $C$ can only be completed if the gadget was in normal operation before checking.
Conversely, the sequence of pushes described above completes $C$ from any normal-operation state.
Thus $\pss{G}{C}$ is exactly $G$ restricted to normal operation.
\end{proof}

From Lemmas~\ref{lem:OR_sound}, \ref{lem:OR_complete}, and~\ref{lem:OR_check} it follows that $G$ is a checkable GORP OR gizmo.

\subsection{Crossover}

\begin{figure}
\centering
\subcaptionbox{\label{fig:NCL Crossover A} $b$ is open.}{\begin{overpic}[scale=0.4]{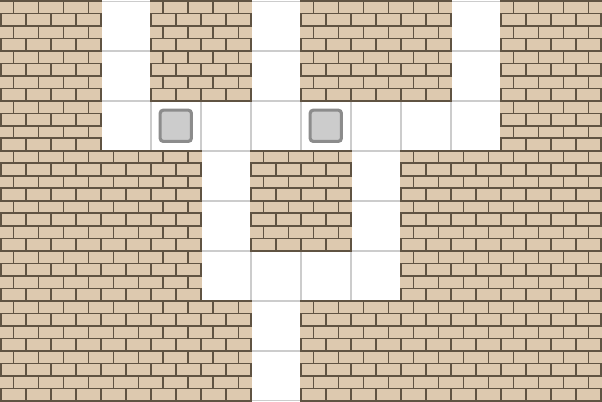}
\put(49,-2){\makebox(0,0)[r]{$x$}}
\put(24,69){\makebox(0,0)[r]{$a$}}
\put(49,69){\makebox(0,0)[r]{$y$}}
\put(82,69){\makebox(0,0)[r]{$b$}}
\end{overpic}}
\hfil
\subcaptionbox{\label{fig:NCL Crossover B} $a$ is open.}{\begin{overpic}[scale=0.4]{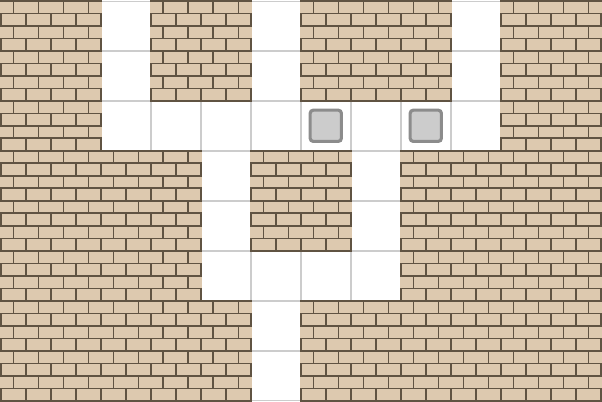}
\put(49,-2){\makebox(0,0)[r]{$x$}}
\put(24,69){\makebox(0,0)[r]{$a$}}
\put(49,69){\makebox(0,0)[r]{$y$}}
\put(82,69){\makebox(0,0)[r]{$b$}}
\end{overpic}}
\hfil\\[0.4cm]
\subcaptionbox{\label{fig:NCL Crossover AB} Neither location is open.}{\begin{overpic}[scale=0.4]{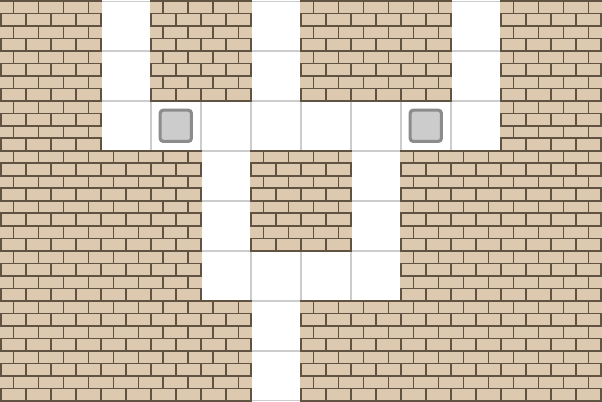}
\put(49,-2){\makebox(0,0)[r]{$x$}}
\put(24,69){\makebox(0,0)[r]{$a$}}
\put(49,69){\makebox(0,0)[r]{$y$}}
\put(82,69){\makebox(0,0)[r]{$b$}}
\end{overpic}}
\hfil
\subcaptionbox{\label{fig:NCL Crossover broken} The broken state.}{\begin{overpic}[scale=0.4]{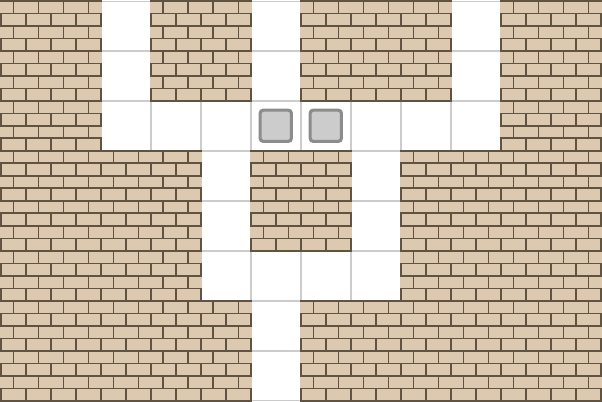}
\put(49,-2){\makebox(0,0)[r]{$x$}}
\put(24,69){\makebox(0,0)[r]{$a$}}
\put(49,69){\makebox(0,0)[r]{$y$}}
\put(82,69){\makebox(0,0)[r]{$b$}}
\end{overpic}}
\caption{Checkable GORP crossover implemented in Push-1, showing only useful states.}
\label{fig:NCL crossover}
\end{figure}

Because Push-1 is inherently planar, we need the planar version of GORP gizmos, which requires a GORP crossover.
Figure~\ref{fig:NCL crossover}
shows our Push-1 crossover gadget in its four useful states.
Any other state is dominated by one of the useful states
in the sense that any traversal sequence that can be performed from that state can also be performed from one of the useful states:
leaving a block next to a corner is better than leaving it in a corner,
leaving a block next to the left or right T junction is better than leaving it in that junction, and
leaving it right of the $y$ T junction is better than leaving it in that junction unless that position is already occupied.
Thus we assume that the gadget remains in one of the useful states,
as there is no reason for the agent to leave the gadget in a different state.

We define the \defn{normal operation} states to be Figures~\ref{fig:NCL Crossover A}, \ref{fig:NCL Crossover B}, and~\ref{fig:NCL Crossover AB};
and the starting state is Figure~\ref{fig:NCL Crossover A}.
In normal operation, the traversals $x \to y$ and $y \to x$ can always be performed without changing the state, as required by Definition~\ref{def:gorp crossover}.

\begin{lemma}\label{lem:crossover_sound}
  In normal operation, the Push-1 crossover gadget is sound with respect to the GORP wire vertex with initial state $\{a\}$.
\end{lemma}
\begin{proof}
  By Lemma~\ref{lem:sound_s_choice}, all we need to show is that $a$ and $b$ can never both be open in normal operation.
  Suppose for contradiction that both $a$ and $b$ were open.
  Then it must be possible to enter at $x$ and leave at either $a$ or $b$.
  The only way that can be the case is if both movable blocks are in the middle, and thus adjacent to each other as in Figure~\ref{fig:NCL Crossover broken}, which contradicts the assumption of normal operation.
\end{proof}

\begin{lemma}\label{lem:crossover_complete}
  In normal operation, the Push-1 crossover gadget is complete with respect to the GORP wire vertex with initial state $\{a\}$.
\end{lemma}
\begin{proof}
  Let $S_i$ be a sequence of valid GORP wire vertex states.
  We use the same argument structure that we used for AND and OR vertices.
  Again, we maintain the additional invariant that every location not in $S_i$ is open after step $i$.
  We just need to show how to open $\ell$ by performing $[x\to x,\ell\to x]$ for $\ell=a$ and $\ell=b$.

  For $\ell = a$, we enter at $x$, move the right block rightward, and exit at $x$. Then we enter at $a$, move the left block rightward, and exit at $x$, leaving the gadget in the state in Figure~\ref{fig:NCL Crossover B}.

  For $\ell = b$, we do a nearly symmetric operation: enter at $x$, move the left block leftward, exit at $x$, and then enter at $b$, move the right block leftward, and exit at $x$, leaving the gadget in the state in Figure~\ref{fig:NCL Crossover A}.
\end{proof}

The checking sequence is the single traversal $x\to y$, which is possible in the three normal operation states but not in the broken one.

\subsection{Checking Framework Base Gadgets}

All that remains is to construct the single-use opening (\SO) and single-use closing (\SC) gadgets in \hbox{Push-1}.
Most of the gadgets needed for this task have already been implemented in \cite[Figures~30--32]{GadgetsChecked_FUN2022},
reproduced here as Figures~\ref{fig:push base gadget states}, \ref{fig:push base gadget impls}, and~\ref{fig:push SO}.
In particular, they implement a dicrumbler and a single-use opening (\SO) gadget in Push-1F.
These constructions also work directly in Push-1 because, by inspection of Figure~\ref{fig:push base gadget impls}, they have the property that every fixed wall is part of a $2\times2$ square of fixed walls, so they can be replaced by movable but effectively fixed blocks.
Their work also builds a merged single-use closing gadget (\MSC); however, we need a stronger (unmerged) single-use closing (\SC) gadget for our checking framework.

\begin{figure}[b]
  \centering
  \def\scale{0.7}
  \subcaptionbox{\label{fig: weak closing state}\centering Weak merged closing}[3.25cm]{\includegraphics[scale=\scale]{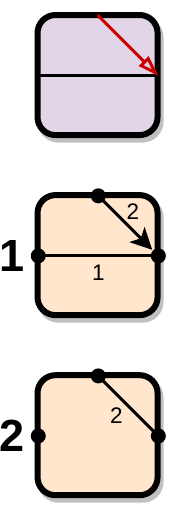}}
  \hspace{1cm}
  \subcaptionbox{\label{fig: no return state}\centering No-return}[3cm]{\includegraphics[scale=\scale]{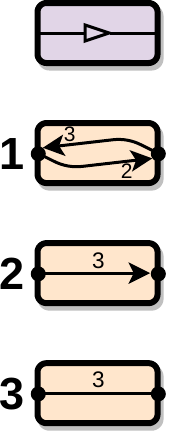}}
  \hspace{1cm}
  \subcaptionbox{\label{fig: weak opening state}\centering Weak opening}{\includegraphics[scale=\scale]{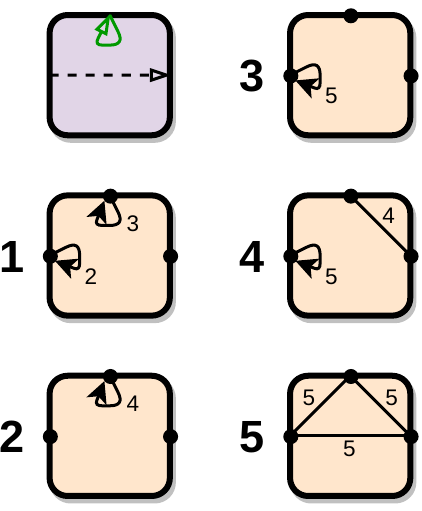}}
  \caption{Icons and state diagrams for Push-1F base gadgets.
    Based on \cite[Figure 30]{GadgetsChecked_FUN2022}.}
  \label{fig:push base gadget states}
\end{figure}

\begin{figure}
\centering
  \def\scale{0.5}
  \subcaptionbox{\label{fig: weak closing push}\centering Weak merged closing}[3.25cm]{\includegraphics[scale=\scale]{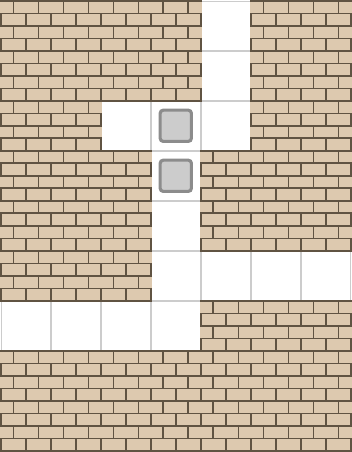}}
  \hspace{1cm}
  \subcaptionbox{\label{fig: no return push}\centering No-return}[6cm]{\includegraphics[scale=\scale]{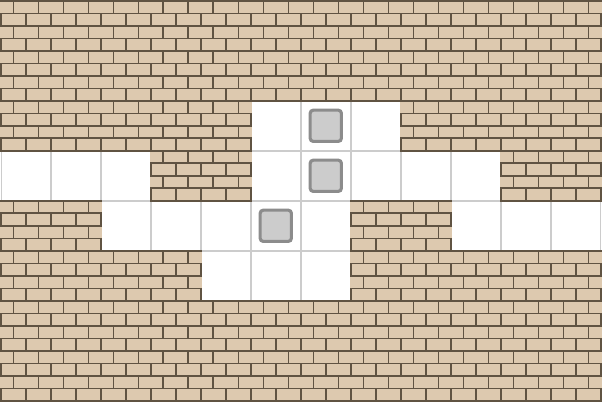}}
  \hspace{1cm}
  \subcaptionbox{\label{fig: weak opening push}\centering Weak opening}[4cm]{\includegraphics[scale=\scale]{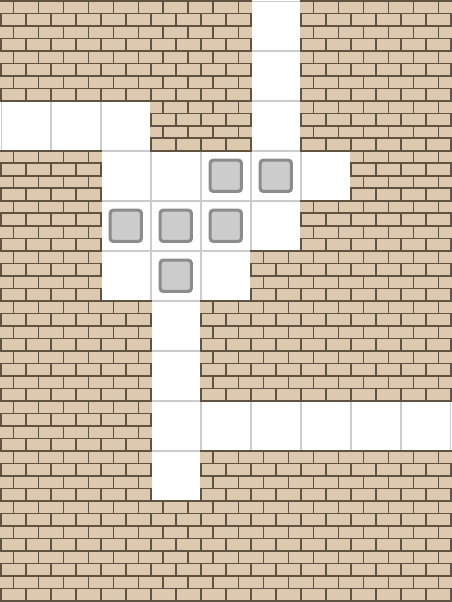}}
  \caption{Push-1(F) implementations of the base gadgets.
    Based on \cite[Figure 31]{GadgetsChecked_FUN2022}.}
  \label{fig:push base gadget impls}
  \end{figure}

\begin{figure}
  \centering
  \def\scale{0.6}
  \subcaptionbox{\centering Dicrumbler}{\includegraphics[scale=\scale]{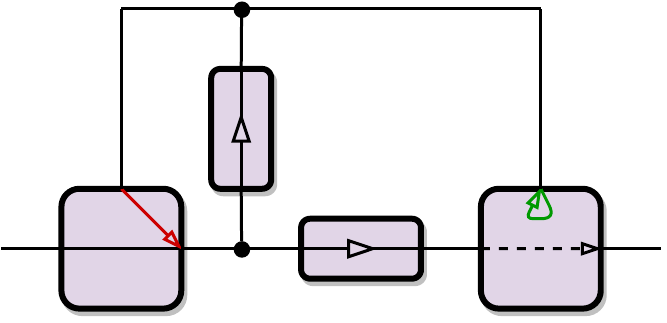}}
  \hfil
  \subcaptionbox{\centering \SO}{\includegraphics[scale=\scale]{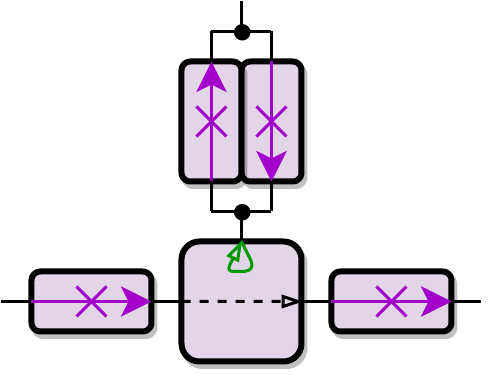}}
  \caption{Constructions of gadgets required for postselection in Push-1F.
    Based on \cite[Figure 32]{GadgetsChecked_FUN2022}.}
  \label{fig:push SO}
\end{figure}

Figure~\ref{fig:SC base simulation} shows our single-use closing gadget.
It consists of three dicrumblers, an \SO{} gadget, and a distant-closing precursor gadget, which is shown in Figure~\ref{fig:SC push}.

\begin{figure}
\centering
\subcaptionbox{
\label{fig:SC push}The distant-closing precursor gadget.}{
\begin{overpic}[scale=0.4]{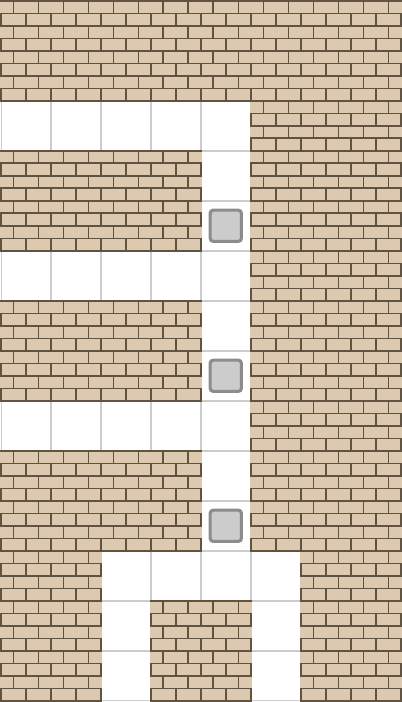}
\put(1,39){\makebox(0,0)[r]{$e$}}
\put(1,61){\makebox(0,0)[r]{$d$}}
\put(1,82){\makebox(0,0)[r]{$c$}}
\put(20,-2){\makebox(0,0)[r]{$x$}}
\put(41.5,-2){\makebox(0,0)[r]{$y$}}
\end{overpic}\vspace{0.75em}}
\hfil
\subcaptionbox{
\label{fig:SC base simulation}
The simulation of an \SC{} gadget built using the dicrumbler and \SO{} from Figure~\ref{fig:push SO} and the distant-closing precursor gadgets from Figure~\ref{fig:SC push}.}{
\begin{overpic}[scale=0.6]{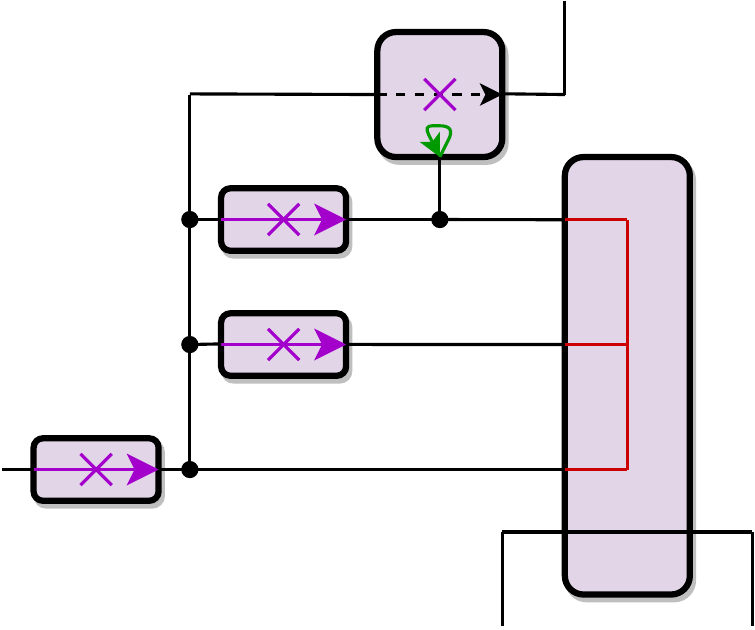}
\put(1,19){\makebox(0,0)[r]{$a$}}
\put(76,81){\makebox(0,0)[l]{$b$}}
\put(65,2){\makebox(0,0)[r]{$x$}}
\put(101,2){\makebox(0,0)[l]{$y$}}
\end{overpic}}
\caption{Our \SC{} gadget. The purple gadget on the right side of (b) is implemented by the gadget in (a).}
\end{figure}

\begin{lemma}\label{lem:SC sim}
The single-use closing gadget in Figure~\ref{fig:SC base simulation} simulates the \SC{} gadget in Figure~\ref{fig:SC_state}.
\end{lemma}
\begin{proof}
In the initial state, the bottom tunnel $x \to y$ can be traversed as much as needed without changing the state of anything. All we need to show is that, upon entering the top tunnel at $a$, the only way the agent can leave is through $b$, and afterwards the gadget is closed and no more traversals are possible.

Suppose the agent enters at $a$, closing the dicrumbler at~$a$.
The agent cannot exit via either of the bottom two locations,
because the bottom block in the distant-closing precursor gadget
can never get above the $e$ tunnel.
Thus the only way to exit is through the top location~$b$.
This can only happen if the agent uses the top dicrumbler,
opens the \SO{} gadget at the top,
and then later makes it back around to the entrance of the \SO{} gadget.
To get back out, the agent needs to push the top movable block in the distant-closing precursor
down far enough to exit at $e$.
To do this, the agent must first enter the distant-closing precursor at $e$ and push the bottom movable block all the way down.
Next, the agent must use the middle dicrumbler to enter the distant-closing precursor at $d$ and push the middle movable block all the way down.
Finally, the agent can use the top dicrumbler, open the \SO{} gadget at the top, and then enter the distant-closing precursor at location $c$.
The agent can then push the top movable block all the way down, leave the distant-closing precursor at $e$, and exit the entire gadget through the top \SO{} gadget.
The key point is that, in order to get back to the \SO{} gadget after opening it, the agent must enter the distant-closing precursor through $c$ and exit it through $e$, which requires all of the movable blocks to be pushed all the way down.

At this point, because there is a block in the $x \to y$ tunnel at the bottom of the distant-closing precursor, that tunnel is impassable.
Also, because of the dicrumbler at entrance $a$ and \SO{} at entrance $b$, neither of these locations can ever be reentered.
Thus, this gadget correctly simulates an \SC{} gadget.
\end{proof}

\section{\PSPACE-Completeness of Push-1}

\begin{theorem}
  Push-1 is \PSPACE-complete.
\end{theorem}

\begin{proof}
  \PSPACE-hardness is largely a matter of assembling the pieces we have collected throughout this paper.
  We use the following chain of reductions:

  \begin{center}
    \begin{tabular}{ccp{.8\linewidth}}
      & & planar NCL with AND and OR vertices \\
      & $\underset{1}{\leadsto}$ &
      planar targeted set reconfiguration with postselections of NCL GORP gizmos and GORP crossover \\
          & $\underset{2}{\leadsto}$ &
          planar reachability with (prefix-closed) NCL GORP gizmos, GORP crossover, single-use opening, and single-use closing \\
          & $\underset{3}{\leadsto}$ &
          planar reachability with checkable NCL GORP gizmos, checkable GORP crossover, weak merged closing, no-return, weak opening, and distant-closing precursor \\
          & $\underset{4}{\leadsto}$ &
          Push-1.
    \end{tabular}
  \end{center}

  The starting NCL problem is \PSPACE-complete by Theorem~\ref{thm:ncl pspace}.
  Reduction 1 is Theorem~\ref{thm:gorp reduction},
  and reduction 2 is via Theorem~\ref{thm:check} and Lemma~\ref{lem:nonlocal set}
  (chaining a constant number of reductions for the nonlocal simulation of each gizmo).

  Reduction 3 is a combination of simulations for the base gadgets
  (Figure~\ref{fig:push SO} and Lemma~\ref{lem:SC sim})
  and nonlocal simulations, again using postselection, for the GORP gizmos.

  Reduction 4 consists of the construction of all of those gadgets in Push-1
  (Figures~\ref{fig:NCL AND}, \ref{fig:NCL OR}, \ref{fig:NCL crossover}, \ref{fig:push base gadget impls}, and~\ref{fig:SC push})
  and the proofs that we have correctly implemented checkable GORP gizmos
  (Lemmas~\ref{lem:AND_sound}, \ref{lem:AND_complete}, \ref{lem:AND_check}, \ref{lem:OR_sound}, \ref{lem:OR_complete}, \ref{lem:OR_check}, \ref{lem:crossover_sound}, and~\ref{lem:crossover_complete}).

  For containment, Push-1 puzzles can easily be simulated in polynomial space, so containment in \NPSPACE${}={}$\PSPACE{} follows from Savitch's Theorem~\cite{Savitch-1970}.
\end{proof}

\section{Open Problems}

One main version of Push remains unsolved in terms of \NP{} vs.\ \PSPACE:
Push-$*$, where the player can push arbitrarily many blocks in one move,
and there are no fixed walls (other than the perimeter of the
rectangular board, which cannot be exited).
Hoffmann \cite{hoffmann-2000-pushstar,demaine2003pushing} proved this
problem \NP-hard, and it seems difficult to make re-usable gadgets
as necessary for \PSPACE-hardness.  Is the problem in \NP?

Many more problems are open for Pull, where the player can pull blocks
instead of pushing them, and for PushPull, where the player can push and
pull blocks.  There are two forms of Pull: Pull!\ requires the player
to pull block(s) whenever they walk away from them, while Pull?\ lets
the player choose whether to walk away or pull.
All hardness results for Pull?\ and Pull!\ \cite{ani2020pspace}
and for PushPull?\ \cite{PRB16} assume fixed walls.
Pulling blocks without fixed walls seems very different
because there is no longer a principle like ``any $2 \times 2$ square of
blocks is effectively fixed'' even with strength~$1$.

Another related problem is $1 \times 1$ Rush Hour, where any block can slide
at any time instead of getting pushed by an agent, and each block can either
only move horizontally or only move vertically.
This problem is known to be \PSPACE-complete,
but only with fixed blocks \cite{RushHour_FUN2020}.
Is it hard without fixed blocks?

Finally, in the context of our checkable gizmos framework,
a natural question is whether all of the auxiliary gadgets are necessary,
or whether this set could be reduced.
This could make it easier to apply the framework in the future.

\section*{Acknowledgments}

This research was initiated during an open problem session that grew out of the
MIT class on Algorithmic Lower Bounds: Fun with Hardness Proofs (6.892)
taught by Erik Demaine in Spring 2019.
We thank the other participants of that class for related discussions
and for providing an inspiring atmosphere.
We also thank the anonymous referees for their helpful comments
which improved the paper.

The figures in this paper were created using SVG Tiler
[\url{https://github.com/edemaine/svgtiler}] and
draw.io [\url{https://github.com/jgraph/drawio}].

\bibliography{main}
\bibliographystyle{alpha-key}

\end{document}